\documentclass[a4paper,onecolumn,11pt,accepted=2023-06-14]{quantumarticle}
\pdfoutput=1
\usepackage[utf8]{inputenc}
\usepackage[english]{babel}
\usepackage[T1]{fontenc}
\usepackage{amsmath}
\usepackage{hyperref}
\usepackage{tikz}
\usepackage{lipsum}
\usepackage{orcidlink}

\usepackage{microtype}
\usepackage{braket}

\usepackage{graphicx}
\graphicspath{ {./images/} }

\usepackage[font={small, it}]{caption}

\usepackage[left=3.5cm, right=3.5cm, top=3.0cm, bottom=3.0cm]{geometry}
\usepackage[shortlabels]{enumitem}

\usepackage{tikz-cd}
\usepackage{systeme}
\usepackage{bm}
\usepackage{tensor}
\usepackage{xparse}
\usepackage{amsmath}
\usepackage{physics}
\usepackage{amsthm}
\usepackage{mathtools}
\usepackage{amssymb}
\usepackage{dsfont}
\usepackage{xcolor}
\usepackage{upgreek}
\usepackage{mathrsfs}
\usepackage[mathscr]{eucal}

\usepackage{hyperref}

\usepackage[parfill]{parskip} 

\usepackage{wrapfig}

\theoremstyle{plain}

\newtheorem{prop}{Proposition}
\newtheorem{lem}[prop]{Lemma}

\theoremstyle{remark}

\newtheorem{ex}{Example}

\renewcommand{\H}{\mathcal{H}}
\newcommand{\B}{\mathcal{B}}
\newcommand{\C}{\mathcal{C}}

\newcommand{\X}{\mathcal{X}}
\newcommand{\V}{\mathcal{V}}

\renewcommand{\S}{\mathcal{S}}

\newcommand{\A}{\mathcal{A}}
\renewcommand{\P}{\mathcal{P}}

\newcommand{\dist}{\operatorname{dist}}
\newcommand{\len}{\operatorname{length}}

\newcommand{\brad}[1]{( #1|}
\newcommand{\ked}[1]{|#1)}
\newcommand{\braked}[2]{( #1 | #2 )}

\DeclarePairedDelimiterX\hej[1]\lbrace\rbrace{#1}

\DeclarePairedDelimiterX{\braketHS}[2]{\langle}{\rangle_{\textsc{hs}}}{#1|#2}
\DeclarePairedDelimiterX{\expect}[1]{\langle}{\rangle}{#1}

\DeclarePairedDelimiterX{\inner}[2]{\langle}{\rangle}{#1, #2}
\DeclarePairedDelimiterX{\iner}[1]{\langle}{\rangle}{#1, #1}

\usepackage[numbers,sort&compress]{natbib}

\title{Geometric Operator Quantum Speed Limit, Wegner Hamiltonian Flow and Operator Growth}
\author{Niklas Hörnedal\,\orcidlink{0000-0002-2005-8694}\,}
\affiliation{Department  of  Physics  and  Materials  Science,  University  of  Luxembourg,  L-1511  Luxembourg, G. D.  Luxembourg}
\author{Nicoletta Carabba\,\orcidlink{0000-0002-9447-9213}\,}
\affiliation{Department  of  Physics  and  Materials  Science,  University  of  Luxembourg,  L-1511  Luxembourg, G. D.  Luxembourg}
\author{Kazutaka Takahashi\,\orcidlink{0000-0001-7321-2571}\,}
\affiliation{Department  of  Physics  and  Materials  Science,  University  of  Luxembourg,  L-1511  Luxembourg, G. D.  Luxembourg}
\affiliation{Department of Physics Engineering, Faculty of
Engineering, Mie University, Mie 514–8507, Japan}
\author{Adolfo del Campo\,\orcidlink{0000-0003-2219-2851}\,}
\affiliation{Department  of  Physics  and  Materials  Science,  University  of  Luxembourg,  L-1511  Luxembourg, G. D.  Luxembourg}
\affiliation{Donostia International Physics Center,  E-20018 San Sebasti\'an, Spain}

\begin{document}

\maketitle
\begin{abstract}
Quantum speed limits (QSLs) provide lower bounds on the minimum time required for a process to unfold by using a distance between quantum states and identifying the speed of evolution or an upper bound to it. We introduce a generalization of QSL to characterize the evolution of a general operator when conjugated by a unitary. The resulting operator QSL (OQSL) admits a geometric interpretation, is shown to be tight, and holds for operator flows induced by arbitrary unitaries, i.e., with time- or parameter-dependent generators. The derived OQSL is applied to the Wegner flow equations in Hamiltonian renormalization group theory and the operator growth quantified by the Krylov complexity.
\end{abstract}

\tableofcontents

In what time scale does a physical process unfold? 
Time-energy uncertainty relations have long been used to estimate characteristic time scales in physical processes, including lifetimes in quantum decay, tunneling times, and the duration of a quantum jump, among others \cite{Pfeifer95,Busch2008,Schulman2008,Dodonov2015}. In the quantum domain, Mandelstam and Tamm put the time-energy uncertainty relation on firm ground in their 1945 work \cite{Mandelstam91}. They   provided its rigorous derivation by combining the Heisenberg equation of motion and the Robertson uncertainty relation. They went a step further by identifying the minimum time for the quantum state of a system to evolve into a distinct state, using the energy dispersion of the initial state as an upper bound to the speed of evolution. In doing so, they introduced an early example of a quantum speed limit (QSL). 

Over the last decades, such an approach has been refined and generalized to a great extent \cite{Deffner17,GongHamazaki22}. Margolus and Levitin found an alternative bound to the speed of evolution in terms of the mean energy of the system \cite{Margolus98}, and ensuing works showed that an infinite family of bounds exist in terms of other moments of the generator of evolution \cite{Zych06,Margolus11}. QSLs have also been derived for time-dependent Hamiltonians \cite{Uhlmann1992,Deffner2013JPA,Okuyama2018}, open systems described by master equations \cite{Taddei13,delcampo13,Deffner2013PRL,Campaioli2019}, and with a stochastic evolution under continuous quantum measurements  \cite{GarciaPintos19}. They have been further extended to the classical domain, with applications ranging from Hamiltonian dynamics to stochastic thermodynamics \cite{Shanahan18,Okuyama18,Shiraishi18,Nicholson2020,VoVuHasegawa20,VuHasegawa21,GarciaPintos2022}.
QSLs generally involve a notion of distance between quantum states and an upper bound to the speed of quantum evolution.
While the Bures angle, defined in terms of the Uhlmann fidelity \cite{Bengtsson2017}, is often the default choice for the distance between quantum states, other alternatives  can provide tighter QSLs \cite{Pires16,Campaioli2018,Campaioli2019,Hornedal22MT}. This freedom is particularly important in the context of many-body systems given the orthogonality catastrophe and the growth of the Hilbert space with the system size \cite{Bukov2019,Fogarty20,Suzuki20,Delcampo21,Hamazaki2022}.

The emphasis on quantum state distinguishability was a key stepping stone in developing QSLs. Today, QSLs find manifold applications in quantum metrology and parameter estimation \cite{Braunstein96,Giovannetti11,Toth14,Beau17}, quantum control \cite{Caneva09,An16,Funo17,Campbell2017}, and quantum thermodynamics \cite{delcampoGoold2014,Binder15}, among other fields.
However, certain phenomena  are naturally described in terms of operator flows, i.e., the continuous evolution of an operator according to given equations of motion. A typical example is the description of quantum evolution in the Heisenberg picture, but the relevance of operator flows is not restricted to quantum dynamics in rotating frames.
Operator flows naturally arise in the Wegner flow equations for  Hamiltonian renormalization \cite{Wegner94,GlazekWilson93,GlazekWilson94,Wegner2001,Kehrein2007}, the study of operator growth and quantum complexity \cite{Keyserlingk18,Nahum18,Rakovszky18,Khemani18,Parker2019}, and correlation functions \cite{Forster18}, to name some examples. 

Motivated by these applications, we have introduced the notion of QSL for operator flows in Ref. \cite{Carabba22}, which we shall term Operator QSL (OQSL) hereafter. Another approach pursuing QSLs for observables was proposed in \cite{Mohan21}.
At variance with conventional QSLs, OQSLs involve a notion of distance between operators (instead of quantum states) and an upper bound on the corresponding speed of evolution. OQSLs have proved useful in characterizing  the time evolution of autocorrelation functions, setting bounds on dynamical susceptibilities arising in linear response theory, and the precision in parameter estimation with thermal quantum systems \cite{Carabba22}. However, in their current form, OQSLs are restricted to flows associated with a constant generator, i.e., one that is  independent of time or the  relevant parameter characterizing the evolution. In addition, OQSLs lack in their present formulation an intuitive geometric description, common in other bounds arising in quantum information geometry, such as the Mandelstam-Tamm QSL. 

In this work, we introduce an OQSL valid under unitary dynamics generated by a time- (or parameter-) dependent generator, that can be a generic operator or an observable, e.g., a Hamiltonian. The resulting bound admits a geometric interpretation and is shown to be tight. We illustrate its usefulness in the study of Wegner flow equations in the theory of Hamiltonian continuous renormalization group \cite{Wegner2001,GlazekWilson93,GlazekWilson94}.  We characterize the OQSL for the Hamiltonian flow  in detail and illustrate the possibility of  saturating it when the flow of Hamiltonian parameters is described by the Toda equations. We further apply the new OQSL to the problem of operator growth in unitary quantum dynamics in Krylov space, i.e., as characterized by the Krylov complexity \cite{Parker2019,Barbon2019,Caputa2019,Dymarsky20,Rabinovici2021,Hornedal22}. We close with a discussion and an outlook pointing out directions for further work.

\section{Motivation}

In quantum physics, a broad class of time-correlation functions is defined through a positive semi-definite inner product. Consider any operator $A$ evolving unitarily according to $\dot{A}_t = i\mathds{L}A_t$, where $\mathds{L} = [H,\cdot]$ is the Liouvillian superoperator. One class of time-correlation functions that have been considered in the literature takes the form $C(t) = \braked{A}{A_t}$, defined with the help of a Hermitian bilinear form in the space of operators,
\begin{equation}
\label{time-cor}
    \braked{A}{B} = \Tr(A^\dagger\rho_1B\rho_2).
\end{equation}
The operators $\rho_1$ and $\rho_2$ are in general positive semi-definite, commute with the Hamiltonian, and need not have unit trace. A familiar example of this type of correlation function is obtained by setting $\rho_1 = \rho_2 = \mathds{1}$.
Eq. \eqref{time-cor} reduces then to the Hilbert-Schmidt inner product. Another familiar instance corresponds to the choice of the identity $\rho_1 = \mathds{1}$ and  the canonical Gibbs state $\rho_2 = e^{-\beta H}/Z$ at inverse temperature $\beta$, with partition function $Z=\Tr e^{-\beta H}$. Note that in general, the bilinear form in \eqref{time-cor} is not positive definite, which means that it does not always define an inner product; instead, it will define a positive semi-definite inner product.\footnote{We use the convention that an inner product must satisfy positive-definiteness, be linear in the second argument and conjugate symmetric} As a result, there might exist cases in which $A\neq 0$ and $\braked{A}{A} = 0$ are simultaneously fulfilled. Another example of a correlation function defined through a positive semi-definite inner product, i.e. $C(t) = \braked{A}{A_t}$, is given by the so-called Kubo inner product \cite{Mori65,Kubo66,Forster18,Recursionmethod}
\begin{equation}
    \braked{A}{B} = \frac{1}{\beta}\int_0^\beta d\lambda \expect{e^{\lambda H}A^\dagger e^{-\lambda H}B}_\beta - \expect{A^\dagger}_\beta\expect{B}_\beta.
\end{equation}
Here, $\expect{\cdot}_\beta$ denote the thermal expectation value and $\beta$ is once again the inverse temperature.

The function $\norm{A} = \sqrt{\braked{A}{A}}$ is a seminorm and we note that the condition $\norm{A_t} = \norm{A}$ is satisfied in the above examples for all $t$.\footnote{A seminorm fulfills all the properties of regular norm except that non-zero vectors can have norm 0.} We might then ask, given a unitary flow induced by a possibly time-dependent Hamiltonian $H$, what is the minimal time $\tau$ for an initial observable $A$ to reach some specific value of $C(t)$ provided that $\norm{A_t} = \norm{A}$ is satisfied during the whole evolution? What follows is a derivation of a speed limit that lower bounds this minimum time.

\section{An Operator Quantum Speed Limit}
\label{OQSL}

The complex Hilbert space $\H$ used to model the system  will be assumed throughout this paper to be finite-dimensional. We define $\B$ to be the space of linear operators acting on this space. 
Moreover, let $\textrm{End}(\B)$ be the space consisting of all vector space endomorphism of $\B$. This linear space is commonly referred to as the Liouville space in the literature \cite{Ernst90,Gyamfi20} and its elements are referred to as superoperators.

Positive semi-definiteness of $\braked{\cdot}{\cdot}$ makes it possible for $C(t)$ to be constant for certain paths $A_t$ even though the operator is changing in time. We will see that there is a way of ``carving away'' the degrees of freedom in $\H$ that do not contribute to changes in $C(t)$. In doing so, we obtain an effective Hilbert subspace $\H_\P$, where the restriction of $\braked{\cdot}{\cdot}$ onto this subspace defines a proper inner-product. We will see that the condition  $\norm{A_t} = \norm{A}$ implies that $A_t$ will be situated on a sphere centered at the origin in $\H_\P$ with radius $\norm*{A}$. This observation will then enable us to derive a geometric OQSL.

\subsection{Construction of the effective Hilbert space}

We  let $\inner{\cdot}{\cdot}_\textsc{h}$ denote the Hilbert-Schmidt inner product on $\B$. It can be shown that any positive semi-definite inner product $\braked{\cdot}{\cdot}$ can be expressed as $\braked{\cdot}{\cdot} = \inner{\cdot}{\P \cdot }_\textsc{h}$, where $\P$ is a positive semi-definite superoperator with respect to the Hilbert-Schmidt inner product. This superoperator will be unique, provided that $\braked{\cdot}{\cdot}$ has been specified---see Appendix \ref{app-A}. As we prove in Appendix \ref{app-B}, a useful relation between $\braked{\cdot}{\cdot}$ and $\P$ is that 
\begin{equation}
\label{kernel}
    \norm{A} = 0 \iff\P A = 0.
\end{equation}
Since $\P$ is self-adjoint, it follows from the spectral theorem that the linear space $\B$ can be expressed as a direct sum of the eigenspaces of $\P$ \cite{Friedberg2018}. Consequently, the image of $\P$ will be spanned by the eigenvectors corresponding to non-zero eigenvalues, and we thus have that $\B = \textrm{im}(\P)\oplus \textrm{ker}(\P)$, where $\textrm{im}(\P)$ and $\textrm{ker}(\P)$ are the image and kernel of $\P$ respectively. Relation \eqref{kernel} then implies that the restriction of $\braked{\cdot}{\cdot}$ to $\textrm{im}(\P)$ will be positive definite and will thus define an inner product on the Hilbert space defined by $\H_\P = \textrm{im}(\P)$. We will express this inner product with the usual bracket notation $\inner{\cdot}{\cdot}$.

For any operator $A\in\B$, we will let $\hat{A}\in\H_\P$ denote the orthogonal projection of $A$ onto $\H_\P$.\footnote{By orthogonal, we mean orthogonality measured by the Hilbert-Schmidt inner product.} More explicitly, given the spectral decomposition $\P = \sum_{k} p_k \Pi_k$, where $p_k$ are the non-zero eigenvalues of $\P$ and $\Pi_k$ are the corresponding eigenprojections, we have that $\hat{A} = \sum_k \Pi_k A$. We show in Appendix \ref{app-C} that
\begin{equation}
    \label{inner-equality}
    \braked{A}{B} = \braket*{\hat{A}}{\hat{B}} \quad \forall A,B\in\B.
\end{equation}
An important consequence of \eqref{inner-equality} is that $C(t) = \braket*{\hat{A}}{\hat{A}_t}$. This means that if we are interested in how $C(t)$ changes over time, then we only need to consider the projected dynamics $\hat{A}(t)$ of $A(t)$.


\subsection{Deriving the speed limit}
\label{QSL derivation}

If $d$ is the complex dimension of the Hilbert space $\H_\P$, then we can view it as a real vector space isomorphic to $\mathds{R}^{2d}$. We endow this space with the Riemannian metric given by the real part of the inner product $\inner{\cdot}{\cdot}$. The condition that $\braked{A}{A} = \braked{A_t}{A_t}$ holds for all $t\in[0,\tau]$ then means that $\hat{A}_t$ will be situated on the $(2d-1)$-dimensional sphere $\S_\norm*{A}$ with radius $\norm*{A}$ centered at the origin. The shortest path connecting two points on the sphere lies on a great circle with a distance given by the angle between the points times the radius. In other words, if $\hat{A}$ and $\hat{B}$ are two operators on the sphere with radius $\norm{A}$, then the geodesic distance is given by
\begin{equation}
    \dist(\hat{A},\hat{B}) = \norm{A}\arccos(\frac{\Re \braked{A}{B}}{\norm{A}^2}).
\end{equation}

The length of any curve $\hat{A}_t$ on the sphere must be greater or equal to the geodesic distance. As a consequence, we obtain an operator quantum speed limit $\tau_\textsc{qsl}$ by noting that
\begin{equation}
\label{first_QSL}
    \tau = \frac{\len([\hat{A}_t])}{\frac{1}{\tau}\len([\hat{A}_t])}\geq \frac{\dist(\hat{A}_0,\hat{A}_\tau)}{\frac{1}{\tau}\len([\hat{A}_t])},
\end{equation}
where $\len([\hat{A}_t]) = \int_0^\tau \norm{\mathds{L}_tA_t}dt$ is the length of the evolution. More explicitly, the speed limit 
can be written in terms of $C(t)$ as
\begin{equation} \label{result1}
    \tau \geq \tau_\textsc{qsl},\quad \tau_\textsc{qsl} = \frac{\sqrt{C(0)}\arccos(\Re\frac{ C(\tau)}{C(0)})}{\V_\tau},
\end{equation}
where $\V_\tau=\frac{1}{\tau}\int_0^\tau \norm{\mathds{L}_tA_t}dt$ is the time averaged speed of the evolution. For non-constant speeds, $\V_\tau$ needs in practice be replaced by an upper bound on the speed of the evolution in order to make $\tau_\textsc{qsl}$ independent on the time $\tau$ one is trying to estimate.

We stress that the speed limit $\tau$ is only guaranteed to hold whenever $\norm{A_t}$ is preserved for the evolving operator. This is always satisfied in the case when $\P = \mathds{1}$ so that $\braked{\cdot}{\cdot}$ becomes the Hilbert-Schmidt inner product. For more general choices of $\P$ however, norm preservation will not be guaranteed to hold for all initial operators. In fact, the norm is preserved for all initial operators if and only if $[\mathds{L},\P] = 0$. Thus, Eq. \eqref{result1} characterizes an infinite set of universal speed limits. It can be applied both for quantum states and observables and is fairly easy to calculate. We will analyze its tightness further down in the text.

Let us express \eqref{result1} in terms of the components of $A$, $\P$ and the eigenvalues of $\mathds{L}$. Assuming $[\mathds{L},\P] = 0$, let $\P_{nm}$ and $\omega_\alpha = E_n - E_m$ be the eigenvalues of $\P$ and $\mathds{L}$ with respect to a common eigenbasis. Moreover, let $A_{ij}$ be the components of $A$ with respect to this basis. Define $v_\alpha = \frac{\P_{nm}\abs{A_{nm}}^2}{C(0)}$, where $C(0) = \sum_{n,m}\P_{nm}\abs{A_{nm}}^2$. We can then express \eqref{result1} more explicitly, as
\begin{equation}
    \tau_\textsc{qsl} = \frac{\arccos(\Re\sum_{\alpha}v_\alpha e^{i\omega_{\alpha}\tau})}{\sqrt{\sum_{\gamma}v_\gamma \omega_{\gamma}^2}},
\end{equation}

where the indices $\alpha$ and $\gamma$ runs from 1 to $\textrm{dim}(\H)^2$.
\begin{ex}
    In the case when $C(t) = \Tr(A^\dagger A_t e^{-\beta H}/Z)$ we have that $\P A = Ae^{-\beta H}/Z$ and a common eigenbasis between $\mathds{L}$ and $\P$ is given by the operators $\ketbra{E_n}{E_m}$ where the eigenvalues of $\P$ are given by $\P_{nm} = e^{-\beta E_m}/Z$.
\end{ex} 

\begin{ex}
    For the Kubo inner-product we can view $\P$ as the composition $\P = \A \circ \B \circ \C$ of the three superoperators defined by $\A(A) = \frac{Ae^{-\beta H}}{Z}$, $\B(A) = \frac{1}{\beta}\int^{\beta}_0 d\lambda e^{-\lambda H} Ae^{\lambda H}$ and $\C A = A - \Tr(A)\mathds{1}$. The superoperators $\A$ and $\B$ commute with a common eigenbasis being given by the eigenvectors $\ketbra{E_n}{E_m}$. A speed limit for a time-evolving operator $A_t$ using the Kubo inner-product can then be obtained by using $\Tilde{A}_t = \C(A_t)$ and $\Tilde{\P} = \A\circ\B$ instead of $A_t$ and $\P$ and then use that $\Tilde{\P}_{nm} = \frac{e^{-\beta E_m}}{\beta Z}\int_0^\beta e^{\lambda(E_m-E_n)}d\lambda$ and $\Tilde{A}_{nm} = A_{nm} - \sum_k A_{kk}$.
\end{ex}

Let us note that the OQSL \eqref{result1} is saturated by any traceless observable involving solely the coupling of the ground state and an excited state of a time-independent Hamiltonian, in the same spirit as the superposition of  these eigenstates saturates the standard QSLs for states \cite{Levitin2009}. To show this, let us consider a time-independent Hamiltonian $H$ and let us denote by $\ket{0}$ and $\ket{E}$ its ground state and the eigenstate with energy $E$: $H\ket{0}=0$ and $H\ket{E}=E\ket{E}$. Then, it is straightforward to see that the operator $A=\frac{1}{\sqrt{2}}(\ketbra{0}{E}+\ketbra{E}{0})$ exactly saturates the OQSL \eqref{result1} with respect to the Hilbert-Schmidt inner product $\braked{A}{B}=\Tr (A^\dagger B)$. Indeed, the autocorrelation function reduces to $C(t)=\cos(Et)$, in natural units, while the velocity is constant and takes the value $\mathcal{V}=\|[H,A]\|=E$. Therefore, being $C(0)=1$ due to the normalization, the OQSL \eqref{result1} reduces to an identity at any time.

This intuition, that when the operator dynamics is confined to two time-independent Hamiltonian eigenspaces the evolution occurs along a geodesic trajectory, can be made more rigorous and general, as we shall describe below identifying the general conditions for the saturation of the OQSL \eqref{result1}.

\subsection{Refined speed Limit}
\label{refined}

Suppose we can find a decomposition of $\hat{A}_t$ such that $\hat{A}_t = S + V_t$, where $S$ and $V_t$ stays orthogonal with respect to the metric $\Re\inner{\cdot}{\cdot}$ throughout the evolution. This is for example natural to consider when there are symmetries present in the system where one could consider $S$ being a symmetry commuting with the Hamiltonian. We then have that the minimal time it takes $V$ to reach the operator $V_\tau$ is smaller or equal to the time it takes $\hat{A}$ to reach $\hat{A}_\tau$. We can thus obtain another speed limit by substituting $C(t)$ with $\braked{V}{V_t}$ in \eqref{result1}. Using that $S$ and $V_t$ are orthogonal, one can write $\Re\braked{V}{V_t} = \Re C(t) -\norm{S}^2$. Our refined speed limit thus takes the form  \begin{equation}
\label{result2}
    \tau \geq \tau_\textsc{ref},\quad \tau_\textsc{ref} = \frac{\sqrt{C(0)-\norm{S}^2}\arccos(\Re\frac{ C(\tau)-\norm{S}^2}{C(0)-\norm{S}^2})}{\V_\tau}.
\end{equation}
This OQSL is a generalization of \eqref{result1} since we can always trivially consider the case when $S = 0$. Also, since $\mathds{1}$ stays invariant throughout the flow generated by the Hamiltonian, we can always consider the choice $S = \braked{\mathds{1}}{A}\mathds{1}$. In the particular case when $\braked{\cdot}{\cdot}$ is the Hilbert-Schmidt inner product, this choice of $S$ reduces \eqref{result2} to the speed limit, denoted by $T_\Theta$, in \cite{Campaioli2018}.

What is worth noting is that the speed limit \eqref{result2} becomes tighter the larger the norm of $S$ is. We can understand this from a geometrical perspective. The speed limit is saturated whenever the traced-out curve of $V_t$ follows a great circle on the sphere $\S_\norm*{V}$. This means that the evolution will be contained in a two-dimensional subspace. If we then let $X$ and $Y$ be a pair of orthonormal vectors spanning this subspace, we must have that $V_t/\norm{V} = \cos{\theta(t)}X + \sin{\theta(t)}Y$, where $\theta(t)$ is some real-valued function of $t$.\footnote{In fact, if we choose $X$ and $Y$ so that $X = V$, then $\theta$ is the angle between $V$ and $V_t$ and is given by $\theta(t) = \arccos(\Re\frac{ C(\tau)-\norm{S}^2}{C(0)-\norm{S}^2})$.} Given that $S$ is non-zero, the corresponding curve of $\hat{A}$ must then be situated on an effective Bloch sphere spanned by the operators $X$, $Y$ and $S$.\footnote{We emphasize that the operators $S$, $X$ and $Y$ are orthogonal with respect to $\Re\inner{\cdot}{\cdot}$ and not necessarily the Hilbert-Schmidt inner product.} More explicitly, we have that $\hat{A} = \norm{V}\cos{\theta(t)}X + \norm{V}\sin{\theta(t)}Y + S$  will trace out a curve following a circle centered at $S$ (see figure \ref{sphere1}). 
\begin{figure}[t]
\centering\includegraphics[width=0.4\columnwidth]{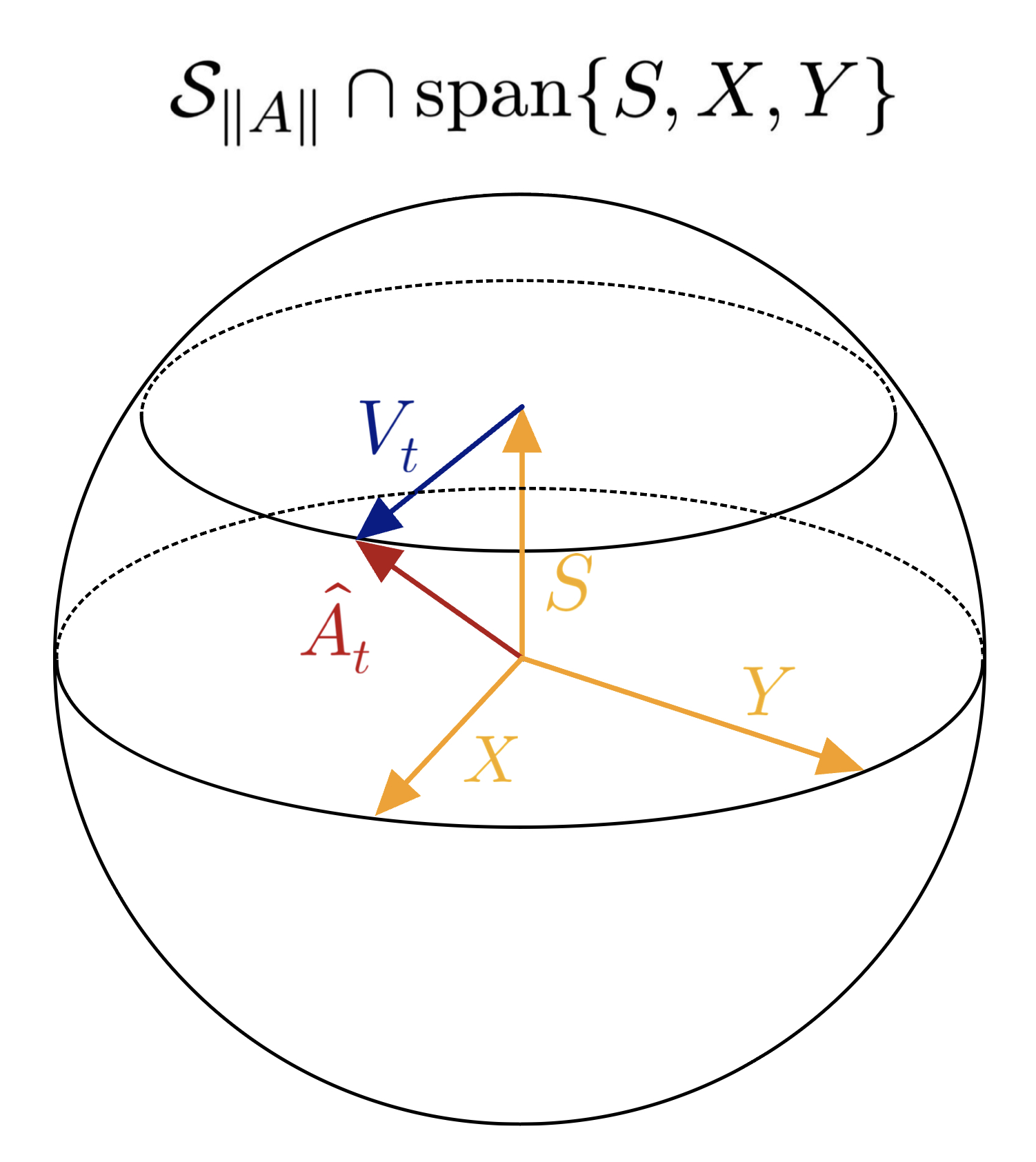}
\caption{
As the evolving operator saturates the refined speed limit, it will move along a circle that is displaced from the origin by $S$. As a consequence, the length of the traced-out curve will be strictly larger than the geodesic distance in $\S_\norm*{A}$. The numerator in \eqref{result2} is exactly the length of this traced-out curve, and we can thus conclude that the refined speed limit is strictly tighter than the original one, given that $S\neq 0$.}
\label{sphere1}
\end{figure}
Since this circle is not centered at the origin, it will not be a great circle on the sphere $\S_\norm*{A}$. A consequence of this is that the length of this curve must be strictly larger than the geodesic distance on $\S_\norm*{A}$. This length is precisely the numerator in \eqref{result2} and we can thus conclude that \eqref{result2} gives a strictly tighter inequality than \eqref{result1}. The difference between these inequalities becomes greater the larger $\norm{S}$ is, which can be seen in figure \ref{sphere2}.
\begin{figure}[t]
\centering\includegraphics[width=0.8\columnwidth]{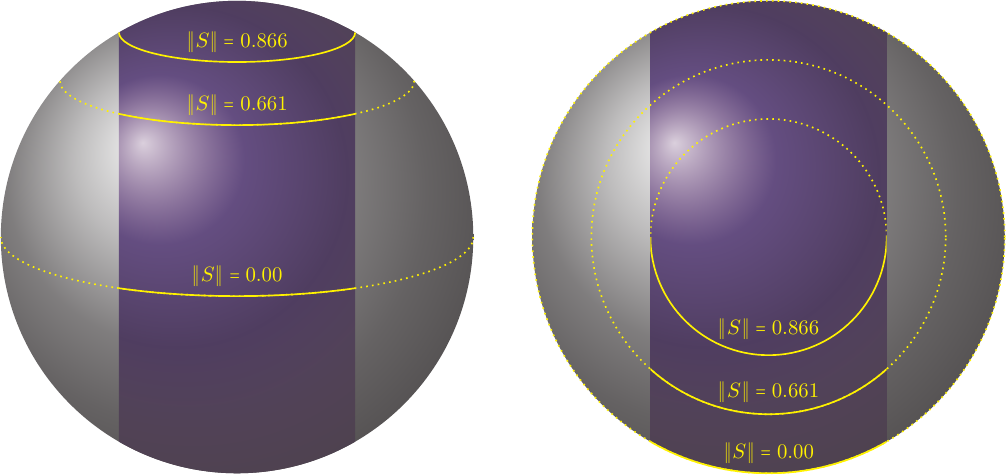}
\caption{
As $\norm{S}$ grows larger, the center of the circle that $\hat{A}_t$ follows will be closer to the poles of the sphere. Consequently, $\hat{A}_t$ moves in a more curved path, as highlighted by the yellow segments where the right figure shows a top view of the sphere, and have to travel further in order to reach the same angle. The result of this is that the refined speed limit becomes increasingly tight the larger $\norm{S}$ is.}
\label{sphere2}
\end{figure}

Consider the case when $\ker(\mathds{L})\cap\H_\P$ is invariant throughout the evolution, for example, when the Hamiltonian commutes with itself for any two points in time. We then have that the norm of $S$ is maximal when $S$ is equal to the orthogonal projection of $\hat{A}$ onto the subspace $\ker(\mathds{L})\cap\H_\P$---see Appendix \ref{app-E}.\footnote{The orthogonal projection here is with respect to the inner product $\inner{\cdot}{\cdot}$.} Calling this component $P_0$, we thus have that the tightest possible refinement, in this case, is given by
\begin{equation}
\label{result3}
    \tau\geq \tau_\textsc{oref}\geq \tau_\textsc{ref}\geq \tau_\textsc{qsl},\quad \tau_\textsc{oref} = \frac{\sqrt{C(0)-\norm{P_0}^2}\arccos(\Re\frac{ C(\tau)-\norm{P_0}^2}{C(0)-\norm{P_0}^2})}{\V_\tau}.
\end{equation}
We will refer to this as the optimal refinement of the OQSL.

\section{More on the Conditions for Saturation} \label{saturation}

As discussed in section \ref{OQSL}, the part of the evolution that induces a change in the time-correlation function will be situated on the sphere $\S_\norm*{A}$ in the effective Hilbert space $\H_\P$. The OQSL \eqref{result1} will then be saturated whenever the traced-out curve follows a great circle. In the case when we could find a decomposition $\hat{A}_t = S + V_t$, where $S$ and $V_t$ remain orthogonal, we could consider the refined speed limit \eqref{result2}, which is saturated if and only if $V_t$ follows a great circle on the sphere $\S_\norm*{V}$. We discussed that for saturation of \eqref{result2}, there exists a pair of orthonormal vectors $X$ and $Y$ such that $\hat{A} = \norm{V}\cos{\theta(t)}X + \norm{V}\sin{\theta(t)}Y + S$. If we choose $X = \hat{A}/\norm{A}$, then the function $\theta(t)$ is the angle between $V$ and $V_t$ with respect to the real-valued inner product $\Re\inner{\cdot}{\cdot}$ and is more explicitly given by $\theta(t) = \arccos(\Re\frac{ C(\tau)-\norm{S}^2}{C(0)-\norm{S}^2})$.

This section will discuss the conditions for saturation for two particular cases. In the first case, we consider the consequences of the operator $\hat{A}$ having support in only two of the eigenspaces of the Hamiltonian. In the second case, we assume $\hat{A}$ to be Hermitian, allowing us to draw connections between the saturation of the OQSLs and the dimension of an underlying Krylov space.

\subsection{Evolution with support in only two eigenspaces of the Hamiltonian}
\label{2-support}

Consider the case when the Hamiltonian commutes with itself at any two points in time. We will consider the case when $\hat{A}$ has non-zero support in only two of the eigenspaces of the Hamiltonian.\footnote{An operator $A$ is said to have non-zero support in a subspace $\X\subseteq\H$ if $\exists \ket{\psi}\in\X$ s.t. $A\ket{\psi}\neq 0$.} Let $E$ and $E'$ denote the energies of these two eigenspaces through time and let $\omega = E - E'$ be the energy gap. If $P_E$ and $P_{E'}$ are the corresponding eigenspace projectors, then $\hat{A} = P_E\hat{A}P_E + P_E\hat{A}P_{E'} + P_{E'}\hat{A}P_E + P_{E'}\hat{A}P_{E'}$ and it is straight forward to check that $\mathds{L} P_{E}\hat{A}P_{E'} = \omega P_{E}\hat{A}P_{E'}$, $\mathds{L} P_{E'}\hat{A}P_{E} = -\omega P_{E'}\hat{A}P_{E}$ and $\mathds{L} P_E\hat{A}P_E = \mathds{L} P_{E'}\hat{A}P_{E'} = 0$. In other words, $\hat{A}$ is spanned by the eigenspaces of the Liouvillian with eigenvalues 0, $\omega$, and $-\omega$. Let $P_0$, $P_{\omega}$, and $P_{-\omega}$ be the corresponding projections of $\hat{A}$ onto these three eigenspaces. More explicitly, we have in this specific case that $P_0 = P_{E}\hat{A}P_{E} + P_{E'}\hat{A}P_{E'}$, $P_\omega = P_{E}\hat{A}P_{E'}$ and $P_{-\omega} = P_{E'}\hat{A}P_{E}$ and we can write $\hat{A} = P_0 + P_{\omega} + P_{-\omega}$. The evolution of $\hat{A}$ is given by
\begin{equation}
\label{optimal dynamics}
    \begin{split}
        \hat{A}_t &= P_0 + e^{i\int_0^t\omega(t') dt'}P_{\omega} + e^{-i\int_0^t\omega(t') dt'}P_{-\omega}\\
        &= P_0 + \cos(\int_0^t\omega(t') dt')\big(P_\omega + P_{-\omega}\big) + \sin(\int_0^t\omega(t') dt')\big(iP_{\omega} - iP_{-\omega}\big)\\
        &= P_0 + \cos{\theta(t)}X_\omega + \sin{\theta(t)}Y_\omega.
    \end{split}
\end{equation}
Here, we have introduced the non normalized operators $X_\omega = P_{\omega}+P_{-\omega}$ and $Y_\omega = iP_{\omega}-iP_{-\omega}$ and  the angle $\theta(t) = \int_0^t\omega(t') dt'$  between $\hat{A}-P_0$ and $\hat{A}_t-P_0$. The requirement that $\norm{A_t}$ is constant in the interval $[0,\tau]$ implies that $P_0$, $X$ and $Y$ are orthogonal and that $X$ and $Y$ have the same norm---see Appendix \ref{app-F}.\footnote{Note that we never need to assume that $\hat{A}$ is Hermitian in order to conclude this.} We can thus conclude that $\hat{A}_t-P_0$ moves along a great arc and thus saturates the optimal refined speed limit \eqref{result3}. One of the consequences of this is that any qubit system with a commutative Hamiltonian must satisfy the speed limit \eqref{result3}.

The above result can be extended to non-commuting Hamiltonians that keep the eigenspaces corresponding to $E$ and $E'$ invariant. One might wonder whether the Hamiltonian must keep these eigenspaces invariant for the evolving operator $\hat{A}_t$ to achieve saturation. The answer is no. To see this, we can consider the commuting Hamiltonian $H_t$ above and add to it any non-trivial Hamiltonian $\Tilde{H}_t$ that commutes with $\hat{A}_t$ for all times $t\in[0,\tau]$. The Hamiltonian $H_t + \Tilde{H}_t$ then generates the same path for $\hat{A}_t$ as the Hamiltonian $H_t$. The difference is that $H_t + \Tilde{H}_t$ will not keep the eigenspaces of $E$ and $E'$ invariant.

\subsection{Relation to Krylov dimension for Hermitian operators}
\label{Krylov}

When $\hat{A}$ is Hermitian, it can be illuminating to describe the saturation conditions in terms of the dimension of the Krylov space of the evolving operator $\hat{A}_t$. The Krylov space is defined to be the smallest subspace containing the evolution. In the Hermitian case, this is a real vector space. We can then conclude that saturation of \eqref{result1} happens if and only if the dimension of the Krylov space is equal to two. Similarly, given that $S\neq 0$, we have that a necessary condition for saturation of \eqref{result1} is that the dimension of the Krylov space is equal to three.

In the case when the Liouvillian is time-independent, the Krylov dimension is given by the number of eigenspaces of $\mathds{L}$ that $\hat{A}$ has support in---see Appendix \ref{app-D}. If one of these eigenspaces is the kernel of $\mathds{L}$, then we can say that the bound will be saturated if and only if the Krylov dimension is smaller or equal to three.

\section{Hamiltonian Flow Equations in Continuous Renormalization Group}

Operator flows are ubiquitous in physics and are not restricted to time evolution.  In this section, we show that the continuous renormalization group provides an arena in which Hamiltonian flows naturally occur and where OQSLs are of relevance. 

As a preamble, we note that the continuous renormalization group is also extensively used in the study of the complexity of quantum states. For instance, the continuous version of the Entanglement Renormalization tensor networks, cMERA \cite{Haegeman13,Nozaki2012}, implements a real space renormalization in which the flow of a quantum state is described as a function of a continuous parameter characterizing the length scale. A cMERA Hamiltonian generates translations of  the cMERA parameter. The use of conventional QSL has been explored in this context to investigate the complexity of states in quantum field theory  \cite{Molina-Vilaplana2018}, using the path integral description of tensor networks \cite{Caputa17,Caputa17b}. These works result from an effort to characterize the growth of quantum complexity in quantum field theories, a context in which conventional QSLs had been  applied  \cite{Brown16,Brown16b}. Indeed, all these results concern flows of quantum states and can thus be tackled with conventional QSL.

In this section, we consider  a different framework for the continuous renormalization group as a paradigmatic example and test bed for our result, the OQSL \eqref{result1}. Specifically, we focus on the Hamiltonian flow formulated by Wegner \cite{Wegner94} and Glazek and Wilson \cite{GlazekWilson93,GlazekWilson94}, as a  method for Hamiltonian block-diagonalization \cite{Wegner2001,Kehrein2007}. Consider the Hamiltonian flow $H(l)=U(l)H(0)U(l)^\dagger$ with respect to some parameter $l$, where $U(l)$ is a unitary operator satisfying $U(0)=\mathds{1}$, and $H(0)$ is the  Hamiltonian of which the block diagonal form is desired. The flowing Hamiltonian satisfies the differential equation
\begin{equation}\label{Wegnerflow}
	\frac{dH(l)}{dl}=[\eta(l),H(l)],
\end{equation}
where $\eta(l)=\frac{dU(l)}{dl}U(l)^\dagger$ is the $l$-dependent generator of the unitary flow. For specific choices of $\eta(l)$, the initial Hamiltonian $H(0)$ eventually flows to its (block)-diagonal form, thus achieving the desired diagonalization. At each point of the flow, let us define the \textit{target} Hamiltonian $H_T(l)$ as the diagonal part of $H(l)$, 
\begin{equation} \label{targetH}
	H_T(l)=\sum_n\epsilon_n(l) \ket{n}\bra{n},
\end{equation}
where we have defined $\epsilon_n(l)\equiv H_{nn}(l)$. One instance of a generator that achieves this is $\eta=[H_T,H]$, as originally proposed in \cite{Wegner94,Wegner2001}. We shall refer to this specific choice as the Wegner flow, for simplicity, even when any Hamiltonian flow described by (\ref{Wegnerflow}) and converging to $H_T(l)$ is generally referred to as a Wegner flow in the literature. Clearly, the choice $\eta=[H_T,H]$  is not the only possibility and we shall consider a different choice  in an example below, that associated with the so-called Toda flow. Let us stress that the $l$-dependent diagonal entries $\epsilon_n(l)$ in Eq.~\eqref{targetH} do \textit{not} correspond to the Hamiltonian eigenvalues unless we have reached the end of the flow $l=l_f$. In that case, the flowing Hamiltonian has been transformed into its diagonal form to coincide with the target part, $H(l_f)=H_T(l_f)$. The final $l_f$ is typically reached as $l\to\infty$, as shown explicitly in the practical example below. As a result,  implementations of the Hamiltonian flow generally terminate at a finite value before the asymptotic diagonal Hamiltonian is exactly reached. This underlines the importance of exploring the Hamiltonian flow along the process, e.g., as characterized by the OQSL.

By applying our general result \eqref{result1} to the evolution \eqref{Wegnerflow}, with $t\to l$, $A_t\to H(l)$ and $H(t)\to -i\eta(l)$, we obtain the following OQSL on the Hamiltonian flow
\begin{equation}\label{QSLWegner}
l \geq l_\textsc{qsl},\quad l_\textsc{qsl} = \frac{\|H\|\arccos{\frac{\braked{H(0)}{H(l)}}{\norm{H}^2}}}{\mathcal{V}_{l}},
\end{equation}
where the speed, averaged over the interval $[0,l]$, is given by
\begin{equation}
    \mathcal{V}_{l} = \frac{1}{l}\int_{0}^{l} \norm{[\eta(t),H(t)]} dt.
\end{equation}
For the rest of this section, we will only consider the case when $\braked{\cdot}{\cdot}$ is the Hilbert-Schmidt inner product. In this case we have that $\norm{[\eta,H]} = \Tr([\eta,H]^2)$.
We first characterize the Hamiltonian flow equations with the conventional choice of the generator due to Wegner, bringing out a formal analogy with dephasing in open systems and pointing out the differences. Next, we introduce an alternative choice of the generator that gives rise to the Toda flow. We then show that, while the Wegner flow does not saturate the OQSL, the Toda flow can do so under the given conditions that we identify.

\subsection{Dephasing-like Wegner flow}

Before applying the OQSL \eqref{QSLWegner} in an explicit example, we show that the Wegner flow, although unitary, features formal similarities with a dephasing evolution, under which the off-diagonal entries of the density matrix $\rho(t)$ decay and the fixed point is set by its diagonal part only. The crucial difference is that in the Wegner flow, the diagonal part evolves as well, and it does so in such a way that the total evolution is unitary and the final diagonal form coincides with the diagonalized initial Hamiltonian. The equation of motion governing the decay of the off-diagonal elements of the Hamiltonian is well-known \cite{Wegner2001,Kehrein2007}.
Here we aim to characterize the decay of the off-diagonal elements $H_{nm}(l)$ with $n\neq m$ and reveal its formal analogy with dephasing. To this end, let us adopt the approach of vectorization \cite{Uzdin2016} and represent an operator $A=\sum_{n,m}A_{nm}\ket{n}\bra{m}$ as a normalized vector
\begin{equation}\label{vector}
	\ket{A}=\frac{1}{\|A\|}\sum_{n,m}A_{nm}\ket{n,m},
\end{equation}
where $|n,m\rangle=\ket{n}\otimes\ket{m}$ and the normalization factor has been introduced to enhance the comparison with the evolution of a quantum state $\rho$. Now, since we are interested in the decay of the off-diagonal elements, let us introduce the (super)projector $\mathds{Q}$ over the non-diagonal part of the Hamiltonian, 
\begin{equation} \label{super-projector}
	\mathds{Q}=\mathds{1}_{\mathcal{B}}-\mathds{P}_{H_T},
\end{equation}
where $\mathds{1}_{\mathcal{B}}$ is the identity on the operator space $\mathcal{B}$ and $\mathds{P}_{H_T}$ is the (super)projector over the target Hamiltonian
\begin{equation} \label{projector-target}
	\mathds{P}_{H_T}=\ketbra{H_T}{H_T}=\frac{1}{\|H_T\|^2}\sum_{n,m}\epsilon_n\epsilon_m\ket{n,n}\bra{m,m}.
\end{equation}
From this expression, it follows that $\mathds{P}_{H_T}\ket{H_T}=\ket{H_T}$ and $\mathds{Q}\ket{H_T}=0$. To characterize the dephasing-like decay realized by the Wegner flow, let us consider the overlap between the flowing Hamiltonian and its non-diagonal part
\begin{equation}\label{defAq}
	\mathcal{A}_\mathds{Q}(l)\equiv \bra{H(l)}\mathds{Q}\ket{H(l)},
\end{equation}
which quantifies how far $H(l)$ is from being diagonal and vanishes as $l\to l_f$, that is as $H(l)\to H_T(l)$. By substituting Eqs.~\eqref{super-projector} and \eqref{projector-target} we obtain
\begin{equation}
	\mathcal{A}_\mathds{Q}(l)=1-\frac{\|H_T\|^2}{\|H\|^2},
\end{equation}
where we stress that $\|H_T\|^2=\sum_n\epsilon^2_n(l)$ depends on $l$. The quantity $\mathcal{A}_\mathds{Q}$ identically vanishes at the end of the flow, when the Hamiltonian is diagonalized and $H(l)=H_T(l)$. Remarkably, if we choose the Wegner generator to be $\eta=[H_T,H]$ \cite{Wegner2001}, then $\mathcal{A}_\mathds{Q}$ decays monotonically. Indeed, its decay rate reads as
\begin{equation}
	\frac{d}{dl}\mathcal{A}_\mathds{Q}(l)=-\frac{1}{\|H\|^2}\frac{d}{dl}\sum_n\epsilon^2_n=\frac{1}{\|H\|^2}\frac{d}{dl}\sum_{n\neq m}|H_{nm}|^2,
\end{equation}
where the last inequality follows from the conservation of the total norm, $\frac{d}{dl}\|H\|=0$. If $\eta=[H_T,H]$ \cite{Wegner2001}, it is straightforward to compute that
\begin{equation}
	\frac{d H_{nm}}{dl}=\sum_k (\epsilon_n + \epsilon_m- 2\epsilon_k) H_{nk}H_{km}
\end{equation}
and therefore
\begin{equation}\label{decoffd}
\frac{d}{dl}\mathcal{A}_\mathds{Q}(l)= - \frac{2}{\|H\|^2} \sum_{n,m}(\epsilon_n-\epsilon_m)^2 |H_{nm}|^2 <0.
\end{equation}
As a result, the overlap \eqref{defAq}, quantifying how far we are from the target, is monotonically decreasing during the flow. This feature suggests an analogy with the well-known model of dephasing, where the purity is found to decrease monotonically \cite{Lidar2006}, in a similar manner as in Eq.~\eqref{decoffd}. To show this, let us consider the case of pure dephasing
\begin{equation}
\frac{d\rho}{dt} = -[X,[X,\rho]],
\end{equation}
where $X$ is a time-independent Hermitian operator satisfying $X\ket{n} = x_n \ket{n}$. Then it can be shown \cite{Lidar2006} that the purity $\Tr \rho^2$ decreases monotonically with time as
\begin{equation} \label{purity-decay}
\frac{d\Tr\rho^2}{dt} = -2\sum_{nm} (x_n-x_m)^2 |\rho_{mn} |^2.
\end{equation}
Now, Eq.~\eqref{decoffd} can be also expressed as a ``purity''-decay of the off-diagonal Hamiltonian $H_\text{off-diag}(l)=H(l)-H_T(l)$ and it reads as
\begin{equation} \label{off-diag-decay}
\frac{d\Tr H_\text{off-diag}^2}{dl} = \frac{d}{dl}\sum_{i\neq j}|H_{ij}|^2=- 2\sum_{i,j}(\epsilon_i-\epsilon_j)^2 |({H_\text{off-diag}})_{ij}|^2,
\end{equation}
which formally corresponds to Eq.~\eqref{purity-decay} with $\rho\to H_\text{off-diag}$ and $X\to H_T$. We conclude that the Wegner flow \eqref{Wegnerflow} generated by $\eta=[H_T,H]$,
\begin{equation}
	\frac{dH}{dl}=[[H_T,H],H],
\end{equation}
which is unitary, suppresses the off-diagonal part of the Hamiltonian as if it was undergoing a dephasing evolution under the action of the diagonal part.

\subsection{Wegner and Toda flows}

As already advanced, there are several choices of the generator $\eta$ 
for diagonalizing a given $N\times N$ matrix $H(0)$.
As a possible form, consider 
\begin{equation}
\label{Tadaeta}
  \eta_{nm}(l)=H_{nm}(l){\rm sgn}\,(m-n),
\end{equation}
for $m\ne n$ and $\eta_{nn}=0$.
Equation (\ref{Wegnerflow}) then reduces to
\begin{equation}
  \frac{dH_{nm}(l)}{dl}=\sum_k H_{nk}(l)H_{km}(l)
  \left[{\rm sgn}\,(k-n)-{\rm sgn}\,(m-k)\right].
\end{equation}
Further, assume that the matrix $H(l)$ takes a symmetric
tridiagonal form.
Then, the equations for diagonal and off-diagonal components
are written respectively as 
\begin{eqnarray}
  && \frac{dH_{nn}(l)}{dl}=2(H_{n,n+1}^2(l)-H_{n-1,n}^2(l))
  \qquad (n=1,2,\dots,N), \label{toda1}\\
  && \frac{dH_{n,n+1}(l)}{dl}=H_{n,n+1}(l)(H_{n+1,n+1}(l)-H_{nn}(l))
  \qquad (n=1,2,\dots,N-1), \label{toda2}
\end{eqnarray}
with $H_{01}=H_{N,N+1}=0$.
This set of equations takes a closed form and is known as
the Toda equations in classical nonlinear integrable
systems~\cite{Toda1967a, Toda1967b, Flaschka1974}. We thus refer to the Hamiltonian flow generated by (\ref{Tadaeta}) as the Toda flow. 

We note that Eq.~(\ref{decoffd}) is not satisfied 
in the present choice of $\eta$.
However, it is still guaranteed that 
the matrix is diagonalized at large $l$
due to the relation 
\begin{equation}
  \frac{d}{dl}\sum_{n=1}^k H_{nn}(l)=2H_{k,k+1}^2(l)\ge 0,
\end{equation}
where $k=1,2,\dots,N$~\cite{Moser1975, Monthus2016}. 
The relation for $k=1$ denotes that $H_{11}(l)$
is a non-decreasing function.
Since $\Tr H^2(l)$ is independent of $l$,
each component of $H(l)$ is not divergent, 
if each component of the original matrix $H_0$ takes a finite value.
We conclude that $\lim_{l\to\infty}H_{11}(l)$ converges to a finite value 
and $\lim_{l\to\infty}H_{12}(l)=0$.
Then, we examine the relation for $k=2$ to
conclude that $\lim_{l\to\infty}H_{22}(l)$ takes a finite value
and $\lim_{l\to\infty}H_{23}(l)=0$.
We can repeat the same consideration for the other values of $k$
to conclude that
$H(l)$ is diagonalized at $l\to\infty$ keeping the eigenvalues
of the matrix unchanged.

A possible realization of the tridiagonal matrix is
the one-dimensional XY model with isotropic interaction \cite{Okuyama2017}, 
\begin{equation}
  H(l)=\frac{1}{2}\sum_{n=1}^{N-1} v_n(l)\left(X_nX_{n+1}+Y_nY_{n+1}\right)
  +\frac{1}{2}\sum_{n=1}^Nh_n(l)Z_n.
\end{equation}
In the $z$-basis, the second term represents the diagonal part 
and the first term represents the off-diagonal part.
The corresponding generator of the time evolution is 
\begin{equation}
  \eta(l)=\frac{i}{2}\sum_{n=1}^{N-1} v_n(l)\left(X_nY_{n+1}-Y_nX_{n+1}\right),
\end{equation}
and the set of coupling functions
$\{v_1(l),v_2(l),\dots,v_N(l),h_1(l),h_2(l),\dots, h_{N-1}(l)\}$
satisfies the Toda equations
\begin{eqnarray}
  && \frac{dh_n(l)}{dl}=2(v_n^2(l)-v_{n-1}^2(l)), \\
  && \frac{dv_n(l)}{dl}=v_n(l)(h_{n+1}(l)-h_{n}(l)).
\end{eqnarray}
This Hamiltonian commutes with the total magnetization
$M=\sum_{n=1}^N Z_n$ and the matrix form of
the single flip sector with $M=\pm (N-2)$
takes a tridiagonal form.

\begin{figure}[t]
\centering\includegraphics[width=0.8\columnwidth]{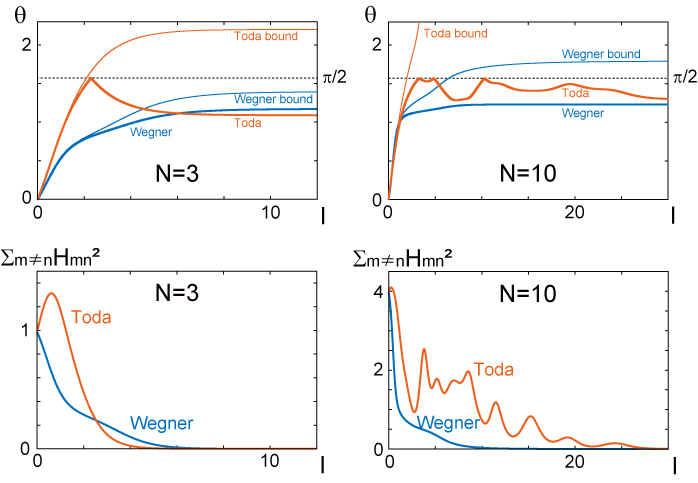}
\caption{
OQSLs for the Wegner and Toda flows.
The upper panels represent 
$\theta(l)=\arccos (\Tr(H(0)H(l))/\norm{H(0)}^2)$ (bold lines)
and its bound $\int_0^l ds\,||[\eta(s),H(s)]||/||H||$ (thin lines)
for Wegner and Toda flows
with $N=3$ (left panel) and $N=10$ (right).
We set the initial matrix as a symmetric tridiagonal form
and each component is taken from a uniform random number between $-1$ and $1$.
In the lower panels, we plot the sum of off-diagonal components 
$\sum_{m\ne n} H_{mn}^2(l)$.
}
\label{wt1}
\end{figure}

We plot examples of the Wegner and Toda flows in Fig.~\ref{wt1}.
We take a traceless symmetric tridiagonal matrix as an initial given Hamiltonian $H(0)$ 
 in which each nonzero component is taken from a uniform random number
between $-1$ and $1$.
The numerical results in Fig.~\ref{wt1} imply that 
the Toda flow gives a tight bound for a small $l$ and becomes worse
for a large $l$ due to the nonmonotonic decay of
the off-diagonal components.

In Fig.~\ref{wt2}, we show how the result is dependent on
the dimension of the matrix $N$. 
For the Wegner flow, when $N$ is not considerably large, 
the angle $\theta(l)$ between $H(0)$ and $H(l)$ grows faster
by the $l$-evolution as $N$ becomes large.
In the Wegner flow, even though we start the time evolution
from a tridiagonal form, 
the matrix breaks the band structure during the flow,
which makes $\theta$ a large value.
It does not necessarily give the property that
the overlap $(H(0)|H(l))$ decays rapidly as a function of $N$,
as we can see in some many-body systems exhibiting the orthogonality catastrophe.
On the other hand, the Toda flow does not show any growing behavior
as a function of $N$.
This is due to the property that the tridiagonal form is kept
throughout the time evolution.
As for $l_\textsc{qsl}$, a saturating behavior is seen for the Wegner flow
and is not seen for the Toda flow.

\begin{figure}[t]
\centering\includegraphics[width=0.8\columnwidth]{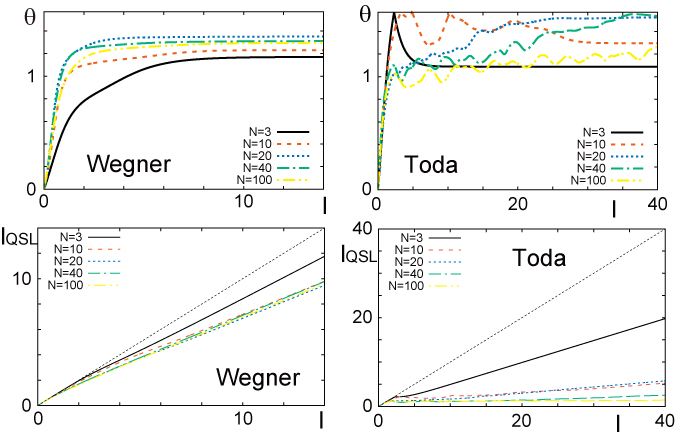}
\caption{
Plot of the $\theta(l)=\arccos (\Tr(H(0)H(l))/\norm{H(0)}^2)$
(top panels) and $l_\textsc{qsl}$ (bottom panels)
for several values of $N$.
We set the initial matrix as a symmetric tridiagonal form
and each component is taken from a uniform random number between $-1$ and $1$.
The left panels represent the Wegner flow, while
the right panels correspond to the Toda flow.
}
\label{wt2}
\end{figure}

The independence of the Toda flow on the matrix dimension implies that we can find a tight OQSL
for specific initial Hamiltonians.
As we have discussed in section \ref{2-support},
saturation is possible when the flowing operator has support
in only two of the eigenspaces of the generator.
In what follows, we will consider the condition that the eigenvectors of the generator $\eta(l)$
are $l$-independent.

Hereafter, we write $H_{nn}(l)=h_n(l)$ and
$H_{n,n+1}(l)=v_n(l)$.
The eigenvalue equation  
$\eta(l)|\varphi\rangle=i\lambda(l)|\varphi\rangle$, 
with a real eigenvalue $\lambda(l)$ and the corresponding
eigenvector $|\varphi\rangle$ is written as 
\begin{equation}
\label{matrix}
\left(\begin{array}{ccccc}
\varphi_2 & 0 &  & &\\
-\varphi_1 & \varphi_3 & 0 &  &\\
0 & -\varphi_2 & \varphi_4 &   &\\
&&\ddots&\ddots&\\
&&0 & -\varphi_{N-2} & \varphi_{N}
\end{array}\right)
\left(\begin{array}{c}
v_1(l) \\ v_2(l) \\ v_3(l) \\ \vdots \\ v_{N-1}(l)
\end{array}\right)
 = i\lambda(l) \left(\begin{array}{c}
\varphi_1 \\ \varphi_2 \\\varphi_3 \\ \vdots \\ \varphi_{N-1}
\end{array}\right),
\end{equation}
and $-v_{N-1}(l)\varphi_{N-1}=i\lambda(l)\varphi_N$,
where $(\varphi_1,\varphi_2,\dots,\varphi_N)$ denotes
the $l$-independent eigenvector $|\varphi\rangle$.
When the diagonal components of the matrix on the left-hand side
are nonzero, the matrix is invertible and
$v_n(l)$ for any index $n$ is proportional
to the same $l$-dependent function $\lambda(l)$.
The possibility that some of the components of $|\varphi\rangle$
are identically zero is excluded since that condition only results in
$|\varphi\rangle=0$.

As an exceptional case,
we can find the eigenvector with $\lambda(l)=0$
when $N$ is odd.
In that case, the eigenvector is written as
\begin{equation}
|\varphi\rangle
\propto\left(1,0,\frac{v_1(l)}{v_2(l)},0,
\frac{v_3(l)v_1(l)}{v_4(l)v_2(l)},0,\dots,0,
\frac{v_{N-2}(l)v_{N-4}(l)\cdots v_1(l)}{v_{N-1}(l)v_{N-3}(l)\cdots v_2(l)}
\right)^{\rm T}.
\end{equation}
The $l$-independence of $|\varphi\rangle$ gives 
the conditions $v_{2k}(l)\propto v_{2k-1}(l)$ with $k=1,2,\cdots, (N-1)/2$.

We note that the eigenvector with $\lambda(l)=0$ is unique if it exists.
Therefore, when we impose the condition that
two of the eigenvectors of $\eta(l)$ are $l$-independent, we are able to invert the matrix in \eqref{matrix} for at least one of the eigenvectors. We can thus conclude in this case that the dependence of $\eta(l)$ on $l$ will be described by a single function $f(l)$, proportional $\lambda(l)$.
Each of the nonzero components is written as
\begin{equation}
 v_n(l)=f(l)v_n. \label{fl}
\end{equation}

We insert the condition (\ref{fl})
into Eqs.~(\ref{toda1}) and (\ref{toda2}) to find 
\begin{eqnarray}
 && h'_n(l)=2f^2(l)(v_n^2-v_{n-1}^2), \label{heq}\\
 && f'(l)=f(l)(h_{n+1}(l)-h_n(l)),  \label{veq}
\end{eqnarray}
where the prime symbol denotes the derivative with respect to $l$.
The second equation (\ref{veq}) shows that
$h_n(l)$ is a linear function in $n$.
Since the constant shift $h_n(l)\to h_n(l)+h_0$
does not change Eqs.~(\ref{toda1}) and (\ref{toda2}),
we set $\sum_{n=1}^N h_n(l)=0$ and obtain 
\begin{equation}
  h_n(l)=\frac{f'(l)}{f(l)}\left(n-\frac{N+1}{2}\right).
\end{equation}
We use this form for the first equation (\ref{heq}).
Then, $v_n^2(l)$ is a quadratic function of $n$ and 
$f(l)$ obeys the differential equation
\begin{equation}
\left(\frac{f'(l)}{f(l)}\right)'=-2d_1f^2(l), \label{fp}
\end{equation}
where $d_1$ represents a constant.
The corresponding form of $v_n$ is 
\begin{equation}
v_n^2=\frac{1}{2}d_1n(N-n)+d_0,
\end{equation}
where $d_0$ represents a constant.

To determine $d_0$, we look at the following condition, which follows from the conservation of the norm 
\begin{equation}
 \sum_{n=1}^N h_n^2(l)+2\sum_{n=1}^{N-1}v_n^2(l)={\rm const}. \label{square}
\end{equation}
Without losing the generality, we can put the form $f(l)=\cos\theta(l)$.
Using Eqs.~(\ref{fp}) and (\ref{square}), we find
as the possible solution 
\begin{eqnarray}
&& d_0=0, \\
&& \left(\frac{\theta'(l)}{\cos\theta(l)}\right)^2=2d_1.
\end{eqnarray}
The $l$ dependence of each component is specified as 
$h_n(l)=h_n\sin\theta(l)$ and $v_n(l)=v_n\cos\theta(l)$.
By representing $d_1$ with respect to $h_1$, we finally obtain 
\begin{eqnarray}
&& h_n(l)=-\frac{2h_1}{N-1}\left(n-\frac{N+1}{2}\right)\sin\theta(l), \\
&& v^2_n(l)=\frac{n(N-n)}{(N-1)^2}h_1^2\cos^2\theta(l),
\end{eqnarray}
and 
\begin{equation}
 \frac{\theta'(l)}{\cos\theta(l)}=\frac{2h_1}{N-1}.
\end{equation}
The differential equation for $\theta(l)$ is easily solved as 
\begin{equation}
  \sin\theta(l)=
  \frac{\sinh\left(\frac{4h_1}{N-1}l\right)
    +\sin\theta(0)\cosh\left(\frac{4h_1}{N-1}l\right)}
       {\cosh\left(\frac{4h_1}{N-1}l\right)
         +\sin\theta(0)\sinh\left(\frac{4h_1}{N-1}l\right)}. 
\end{equation}
In Fig.~\ref{toda},
we plot $\{h_n\}_{n=1,2,\cdots,N}$, $\{v_n\}_{n=1,2,\cdots,N-1}$,
and $\theta(l)$ for $N=20$.

\begin{figure}[t]
    \centering\includegraphics[width=0.8\columnwidth]{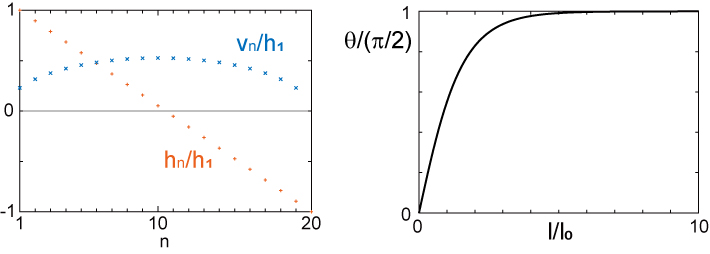}
    \caption{
    Left: The parameters of the initial matrix $H$,
    $\{h_n\}_{n=1,2,\cdots,N}$ and $\{v_n\}_{n=1,2,\cdots,N-1}$,
    resulting in the tight bound in the Toda flow with $N=20$.
    Right: $\theta(l)=\arccos (\Tr(H(0)H(l))/\norm{H(0)}^2)$.
    We set $l_0=(N-1)/4h_1$ and $\sin\theta(0)=0$.}
    \label{toda}
\end{figure}

All components of the matrix $H(l)$ are parameterized by
a single $l$-dependent function $\theta(l)$,
which implies that the time evolution can be denoted by a motion
along an arc in the Bloch space.
In fact, we find that $\theta(l)$ denotes the operator angle
and the dynamics gives the tight bound: 
\begin{equation}
 {\rm arccos}\,\left|\frac{\Tr (H(l)H(0))}{\Tr ( H(0)^2)}\right|
 = \int_0^ldt\,\sqrt{
   \frac{\Tr ([\eta(t),H(t)]^2)}{\Tr (H(0)^2)}}
 = \theta(l)-\theta(0). 
\end{equation}

\section{Operator Growth and Krylov Complexity}

In the previous section, we have considered the flow of an observable, the Hamiltonian, with respect to a parameter different than time. Here we illustrate another application of the OQSL, pointing out that operator flows need not necessarily concern an observable. In particular, given a Liouvillian operator $\mathds{L} = [H,\cdot]$, we show that the geometrical OQSL \eqref{result1} can also be applied to the unitary flow of a superoperator, generated under the action of $\mathds{S}=[\mathds{L}, \cdot]$, which can be accordingly viewed as a ``super Liouvillian''. This kind of flow arises naturally in characterizing the complexity of a given quantum evolution. Specifically, in the context of operator growth, the notion of Krylov complexity 
\cite{Parker2019,Barbon2019,Caputa2019,Dymarsky20,Rabinovici2021,Hornedal22}
 has recently gained attention as a measure of operator complexity for the Heisenberg evolution of an observable under the action of a time-independent Hamiltonian. The evolution of simple, local observables into increasingly complex and nonlocal ones can be described as the operator spreading in the so-called Krylov space. As mentioned in section \ref{Krylov}, the latter provides the minimal subspace in which the Heisenberg dynamics unfolds and is uniquely determined by the Hamiltonian of the system and the initial operator $O_0$. Krylov complexity can then be understood as the mean position of the evolving operator $O_t$ in the so-called Krylov basis. It can be expressed as an expectation value $\braked{O_t}{\mathds{K}O_t}$ of a corresponding (super)operator $\mathds{K}$, known as the complexity operator. At this point, one can again change representation and let the superoperators, such as $\mathds{K}$, evolve while keeping the observable $\ked{O}$ fixed. We shall call this representation the super-Heisenberg picture from the evident analogy with the standard Heisenberg representation of the quantum evolution in the Hilbert space. In this picture, the complexity operator evolves accordingly to the equation
\begin{equation} \label{Kflow}
    \dot{\mathds{K}}=i[\mathds{L},\mathds{K}]
\end{equation}
and the corresponding unitary flow is constrained by the speed limit \eqref{result1}, upon identifying $A$ and $H$ with $\mathds{K}$ and $\mathds{L}$, respectively. We note that our result \eqref{result1} holds for finite dimensions and that the dimension of the Krylov space is always finite whenever the Hilbert space that the observables are defined over is finite \cite{Rabinovici2021}. 

\subsection{Quantum dynamics in Krylov space}

Let us start by briefly recalling how the Krylov space and the corresponding notion of complexity are constructed. For a more detailed discussion, we refer to \cite{Parker2019,Barbon2019,Caputa2019,Dymarsky20,Rabinovici2021,Hornedal22}. The evolution in the Heisenberg picture of an operator $O_t=e^{iHt}O_0e^{-iHt}$ can be formally written in terms of the nested commutators with the Hamiltonian $H$, that is, the powers of the Liouvillian $\mathds{L} = [H,\cdot]$, as $O_t = \sum_{n=0}^\infty \frac{(it)^n}{n!}\mathds{L}^nO$. The space explored during this evolution is given by the span of the infinite set $\{\mathds{L}^nO\}_{n=0}^\infty$ and is precisely the Krylov space. From this infinite set, one can extract an orthonormal, finite basis $\{O_n\}_{n=0}^{D-1}$ by applying the so-called Lanczos algorithm. The first element of the basis coincides with the initial operator $O_0$, which we will assume to be normalized to one. Then, at each iterative step, the next orthogonal vector is constructed as $\ked{A_{n+1}} = \mathds{L}\ked{O_n}-b_{n}\ked{O_{n-1}}$, where $b_n = \norm{A_n}$ is the $n$-th Lanczos coefficient, and the corresponding element of the Krylov basis $\ked{O_{n+1}}$ is obtained upon normalization as $O_n=A_n/b_n$. Throughout this section, we will use the Hilbert-Schmidt inner product $\braked{A}{B}=\Tr A^\dagger B$ between operators. By making use of the Krylov space, the unitary evolution of the operator $O_t$ is effectively mapped to a hopping problem on the one-dimensional, semi-infinite chain represented by the Krylov basis $\{O_n\}_{n=0}^{D-1}$, where the Lanczos coefficients $b_n$ play the role of hopping parameters and the Liouvillian, which takes the tridiagonal form
\begin{equation}\label{LiouvK}
    \mathds{L}=\sum_{n=0}^{D-1} b_{n+1}\ked{O_{n+1}}\brad{O_n}+b_n\ked{O_{n-1}}\brad{O_{n}},
\end{equation}
with $\ked{O_{-1}}=\ked{O_D}=0$, acts as an Hamiltonian for the so-called operator wavefunction $\ked{O_t}$. The Krylov complexity operator $\mathds{K}$ is then defined as the position operator
\begin{equation} \label{K}
    \mathds{K}=\sum_{n=0}^{D-1} n \ked{O_n}\brad{O_n},
\end{equation}
on this lattice. The most studied object in this context is the expectation value of the above (super)-operator with respect to $O_t$ and is known as the Krylov complexity: $K=\braked{O_t}{\mathds{K}O_t}$. Its rate of growth is constrained by the speed limit
\begin{equation} \label{dispersion-bound}
    |\partial_t K(t)|\leq 2b_1\Delta\mathds{K},
\end{equation}
introduced in \cite{Hornedal22} and known as the \textit{dispersion bound}, given that $(\Delta\mathds{K})^2=\braked{O_t}{\mathds{K}^2O_t}-\braked{O_t}{\mathds{K}O_t}^2$ is the variance of $\mathds{K}$ with respect to $O_t$. As was shown in \cite{Hornedal22}, the dispersion bound is saturated at any time if and only if the structure of Krylov space features the so-called \textit{complexity algebra} \eqref{algebra closure}, which we shall introduce below.

It was pointed out in \cite{Caputa2019} that the Liouvillian in Krylov space, given by Eq.~\eqref{LiouvK}, can be written as the sum $\mathds{L}=\mathds{L}_+ + \mathds{L}_-$ of raising and lowering operators that act on the Krylov basis as $\mathds{L}_+\ked{O_n} = b_{n+1}\ked{O_{n+1}}$ and $\mathds{L}_-\ked{O_n} = b_{n}\ked{O_{n-1}}$, respectively. It appears then natural to introduce a super-operator $\mathds{B}= \mathds{L}_+ - \mathds{L}_-$, conjugated to the Liouvillian, and consider their commutator $\Tilde{\mathds{K}} = [\mathds{L}, \mathds{B}]$ \cite{Caputa2019}. The dispersion bound \eqref{dispersion-bound} is identically saturated if and only if these three operators close an algebra, which can only take the form \cite{Hornedal22}
\begin{equation} \label{algebra closure}
        [\mathds{L}, \mathds{B}] = \Tilde{\mathds{K}},\quad [\Tilde{\mathds{K}}, \mathds{L}] = \alpha\mathds{B},\quad[\Tilde{\mathds{K}}, \mathds{B}] = \alpha\mathds{L}
\end{equation}
and implies also that $\Tilde{\mathds{K}} = \alpha\mathds{K}+\gamma\mathds{1}$, where $\alpha,\gamma\in\mathbb{R}$ \cite{Caputa2019}. The Krylov complexity growth rate is then maximal \cite{Hornedal22}. It can be shown that $\gamma$ is always a positive number and $\alpha$ is a real number satisfying the condition $\alpha = -\frac{2\gamma}{D-1}$ for finite Krylov dimension $D$ and $\alpha\geq0$ if $D=\infty$ \cite{Hornedal22}. Moreover, the algebraic closure \eqref{algebra closure}, or equivalently the saturation of the bound \eqref{dispersion-bound}, implies that the Lanczos coefficients evolve according to \cite{Caputa2019,Hornedal22}
\begin{equation} \label{bn}
b_n = \sqrt{\frac{1}{4}\alpha n(n-1) + \frac{1}{2}\gamma n}.
\end{equation}
For $\alpha>0$, this dependence captures the asymptotic linear growth $b_n=\sqrt{\alpha} n$ conjectured by Parker {\it et al.} to hold in generic non-integrable systems, leading to the maximal growth of Krylov complexity \cite{Parker2019}. A paradigmatic example of this class of systems is the celebrated Sachdev-Ye-Kitaev (SYK) model \cite{Chowdhury22}.

\subsection{Saturation of the OQSL by the complexity algebras}

In the ``super'' Heisenberg picture, the observable operator is kept fixed while the complexity operator $\mathds{K}$ evolves unitarily according to 
\begin{equation} \label{Kt}
    \mathds{K}_t = e^{-i\mathds{L}t} \mathds{K}_0 e^{i\mathds{L}t}= \sum_{n=0}^\infty \frac{(-i)^n}{n!} \mathds{S}^n(\mathds{K}_0) \, t^n,
\end{equation}
where $\mathds{S}=[\mathds{L}, \cdot]$ will be referred to as super Liouvillian. In the case of closed complexity algebras, thanks to the commutation relations \eqref{algebra closure}, all the powers $\mathds{S}^n(\mathds{K}_0)$ reduce to terms proportional to either $\mathds{K}_0$ or $\mathds{B}$. In particular, if $\alpha \neq 0$ one can show that \cite{Hornedal22}
\begin{eqnarray} \label{Spower}
  \mathds{S}^{2n}(\mathds{K}_0) &=& (-1)^n \alpha^{n-1}(\alpha \mathds{K}_0 + \gamma\mathds{1} ),\\
  \label{Spowers}
  \mathds{S}^{2n+1}(\mathds{K}_0) &=& (-1)^{n+1} \alpha^{n} \mathds{B},
\end{eqnarray}
where the first equation clearly holds only for $n>0$, being $\mathds{S}^0(\mathds{K}_0)=\mathds{K}_0$. Conversely, when $\alpha=0$ only $\mathds{S}(\mathds{K}_0)= -\mathds{B}$ and $\mathds{S}^2(\mathds{K}_0)= -\gamma\mathds{1}$ survive, being $\mathds{S}^n(\mathds{K}_0)=0$ for $n>2$. Thus,
\begin{equation}
    \mathds{K}_t= \left\{
    \begin{array}{ll} \cos(\sqrt{-\alpha}\,t)\big(\mathds{K}_0 + \frac{\gamma}{\alpha}\mathds{1}\big) -\frac{\gamma}{\alpha}\mathds{1}+\frac{i}{\sqrt{-\alpha}}\sin(\sqrt{-\alpha}\,t)\mathds{B}& \quad \alpha \neq 0, \\
    \mathds{K}_0+i\mathds{B}t+\frac{\gamma}{2}t^2\mathds{1} & \quad \alpha = 0.
    \end{array}
    \right.
\end{equation}
 Equation \eqref{Kt}, together with Eqs. \eqref{Spower}-\eqref{Spowers}, implies that whenever the dispersion bound \eqref{dispersion-bound} is saturated, the full time-evolution of the Krylov complexity operator $\mathds{K}_t$ must be contained in a $3$-dimensional space, spanned by the identity $\mathds{1}$, the initial complexity $\mathds{K}_0$ and $\mathds{B}=[\mathds{K}_0,\mathds{L}]$:
\begin{equation} \label{KtSpan}
    \mathds{K}_t \in \text{Span}\{\mathds{1},\mathds{K}_0,\mathds{B}\}.
\end{equation}
This space defines the ``super'' Krylov space of Krylov complexity itself and turns out to be dramatically simplified by the assumption of closed complexity algebra. From the discussion in section \ref{Krylov}, we can conclude that $\mathds{K}_t$ having a $3$-dimensional Krylov space is a sufficient condition for it to saturate the refined OQSL \eqref{result3}. We recall that here $A_t$ and $\mathds{L}$ are replaced by $\mathds{K}_t$ and $\mathds{S}$, respectively. Then Eq.~\eqref{KtSpan} implies that $\mathds{K_t}$ has support in only three eigenspaces of $\mathds{S}$, one of them corresponding to the eigenvalue $0$: $\mathds{K}_t=\mathds{P}_0 + \mathds{P}_\omega + \mathds{P}_{-\omega}$. Therefore, we conclude that the super operator $\mathds{K}_t-\mathds{P}_0$ saturates the OQSL \eqref{result3}; that is, it evolves along a geodesic trajectory. In other words, the maximal growth rate of Krylov complexity, leading to the saturation of the dispersion bound \eqref{dispersion-bound}, is equivalent to the geodesic evolution of the Krylov complexity operator, provided that we subtract its stationary component $\mathds{P}_0$. Indeed, we remark that $\mathds{P}_0$ is necessarily different from zero since $\Tr \mathds{K} \neq0$, and must be removed to obtain a tight OQSL. From the explicit computation performed below, we shall conclude that the identity is indeed the only stationary component of the Krylov complexity, i.e., $\mathds{P}_0=\Tr \mathds{K}\frac{\mathds{1}}{\|\mathds{1}\|^2}$. By doing so, we shall also explicitly assess the improvement achieved by replacing the OQSL \eqref{result1} with its refined counterpart \eqref{result3}.

\subsection{Computation of the OQSL for the complexity algebras}

The geometrical OQSL \eqref{result1} for the Krylov complexity reads as
\begin{equation} \label{resultK}
    t \geq  \|\mathds{K}_0\| \frac{\arccos\big( \frac{\braked{\mathds{K}_0}{\mathds{K}_t} }{\|\mathds{K}_0\|^2} \big)}{\|[\mathds{L},\mathds{K}_0]\|},
\end{equation}
where, given a time-independent Liouvillian $\mathds{L}$, the velocity of the flow $\|[\mathds{L},\mathds{K}_0]\|$ is constant. The initial complexity $\mathds{K}_0$ coincides with the usual Krylov complexity operator \eqref{K} in the standard Heisenberg picture. In order to assess its deviation from saturation, we explicitly evaluate the OQSL \eqref{resultK} in the case of the $SU(2)$ complexity algebra, that is, when the complexity growth saturates the dispersion bound \eqref{dispersion-bound} in finite dimension \cite{Hornedal22}.
Let us define $\mathcal{V}_\mathds{K} \equiv \|\mathds{K}_0\|^{-1}\|[\mathds{L},\mathds{K}_0]\|$ as the (normalized) velocity of the complexity flow. By explicit computation in the Krylov basis, one can verify that $\mathcal{V}_\mathds{K}$, being $[\mathds{L},\mathds{K}_0]=-\mathds{B}$ by construction, always reduces to
\begin{equation} \label{velK}
    \mathcal{V}^2_\mathds{K} = \frac{\|\mathds{B}\|^2}{\|\mathds{K}_0\|^2} = \frac{2}{\|\mathds{K}_0\|^2} \sum_{n=1}^{D-1} b_n^2,
\end{equation}
where the norm of the complexity is fixed by the Krylov dimension as
\begin{equation} \label{Knorm}
    \|\mathds{K}_0\|^2 = \sum_{n=0}^{D-1} n^2 = \frac{D(D-1)(2D-1)}{6}.
\end{equation}
The velocity \eqref{velK} of the complexity flow is maximized whenever the Lanczos coefficients growth is maximal, i.e.,~linear in $n$, which is the case for maximally chaotic systems according to the universal growth hypothesis \cite{Parker2019}. Indeed, by focusing on the initial scrambling period and neglecting the role of the following plateau and descent in the $b_n$'s \cite{Rabinovici2021}, we conclude that a sub-polynomial behavior $b_n\propto n^\delta$ with $0<\delta<1$ always leads to a smaller velocity $\mathcal{V}_\mathds{K}$ than the linear growth  $b_n\propto n$. We stress that this observation also holds for infinite-dimensional Krylov spaces. As shown below, the velocity \eqref{velK} remains finite in this limit, and once we fix the proportionality constant, the velocity is maximized by the linear growth of the $b_n$'s. In other words, according to the universal growth hypothesis \cite{Parker2019}, complexity flows at the highest speed in maximally chaotic systems.

Let us now focus on the instances of Krylov dynamics that saturate another notion of the speed limit for operator growth, namely the above-mentioned dispersion bound \eqref{dispersion-bound}. As reviewed above, in such cases, the dynamics of Krylov complexity is determined by an underlying $3$-dimensional algebra \cite{Hornedal22} and the Lanczos coefficients obey Eq.~\eqref{bn}. As a result, the velocity \eqref{velK} of the complexity flow can be expressed as
\begin{equation} \label{velK-alg}
    \mathcal{V}^2_\mathds{K}= \frac{\alpha(D-2) + 3\gamma}{2D-1}.
\end{equation}
Before restricting the analysis to the finite-dimensional case $\alpha<0$, where the geometrical QSL \eqref{result1} can be applied, let us stress that this notion of velocity is well defined also in the limit $D\to\infty$, where $\alpha\geq0$. In particular, if $\alpha>0$, when the Krylov complexity diverges exponentially with time as $K(t)\sim e^{\sqrt{\alpha}t}$ \cite{Caputa2019,Hornedal22}, we obtain that $\mathcal{V}_\mathds{K}\to\sqrt{\alpha/2}$: the speed of the complexity operator flow is proportional to the characteristic time scale of the exponential divergence of the Krylov complexity. Instead, if $\alpha=0$, which leads to a quadratic divergence of $K$ \cite{Caputa2019,Hornedal22}, the above-defined speed of the flow vanishes. This singular behavior can be understood as the QSL \eqref{resultK}, derived under the assumption of finite dimension $D$, may not have a well-defined counterpart for $D\to\infty$. In particular, as we shall see below,  the numerator in Eq.~\eqref{resultK} also vanishes for $\alpha=0$, resulting in an indeterminate form $0/0$. Finally, for the Krylov dynamics to saturate the dispersion bound \eqref{dispersion-bound} over a finite-dimensional space, the underlying complexity algebra must be that of $\textit{SU}(2)$, corresponding to the case $\alpha<0$ \cite{Hornedal22}. In such case, given the condition $b_D=0$, the parameters $\alpha$ and $\gamma$ of Eq.~\eqref{bn} are subject to the further constraint $2\gamma=|\alpha|(D-1)$ \cite{Hornedal22}, which also ensures the expression \eqref{velK-alg} to be positive. Indeed, the velocity of the complexity flow generated by the $\textit{SU}(2)$ algebra reduces to $\mathcal{V}^2_\mathds{K}=|\alpha|(D+1)/[2(2D-1)]$. For this class of models, we compute the QSL \eqref{resultK} exactly, thus establishing a direct comparison between the geometrical OQSL introduced in the present work and the dispersion bound that constraints the growth of Krylov complexity \cite{Hornedal22}.

In addition to the velocity of the flow, the other quantity that characterizes the speed limit is the notion of distance spanned during the evolution, which appears at the numerator of Eqs.~\eqref{result1} and \eqref{resultK} and is given in terms of the autocorrelation function, i.e., the operator overlap. In the case of closed complexity algebras, the autocorrelation function of the complexity
\begin{equation} \label{Kautocorr}
    \braked{\mathds{K}_0}{\mathds{K}_t} = \sum_{n=0}^\infty \frac{(-i)^n}{n!} \braked{\mathds{K}_0}{\mathds{S}^n(\mathds{K}_0)} \, t^n
\end{equation}
can be explicitly evaluated by making use of Eqs.~\eqref{Spower}-\eqref{Spowers}. We note that, since $\braked{\mathds{K}_0}{\mathds{B}}=0$, only the even powers of the super Liouvillian $\mathds{S}$ contribute to the sum in Eq.~\eqref{Kautocorr}. Moreover, the only finite-dimensional case where we can straightforwardly apply the  OQSL \eqref{resultK} is that of the $\textit{SU}(2)$ complexity algebra, i.e., when $\alpha<0$. For this class of models, by substituting Eqs.~\eqref{Spower}-\eqref{Spowers} into the expression \eqref{Kautocorr} and recognizing the Taylor expansion of the cosine, we obtain
\begin{equation} \label{Kautocorr-SU2}
    \braked{\mathds{K}_0}{\mathds{K}_t} = (\|\mathds{K}\|^2 + \frac{\gamma}{\alpha}\Tr\mathds{K}) \cos{\sqrt{|\alpha|}t} - \frac{\gamma}{\alpha}\Tr\mathds{K},
\end{equation}
where $\Tr\mathds{K}=D(D-1)/2$ and $\|\mathds{K}\|^2$ is given by Eq.~\eqref{Knorm}. Therefore, in the case of the $SU(2)$ complexity algebra, the autocorrelation function $\braked{\mathds{K}_0}{\mathds{K}_t}$ oscillates at the frequency $\sqrt{|\alpha|}$, which, we note, is half the frequency of oscillation of the Krylov complexity $K$ itself \cite{Hornedal22}. By substituting Eq.~\eqref{Kautocorr-SU2} into Eq.~\eqref{resultK} and using that $2\gamma=|\alpha|(D-1)$, we rewrite the OQSL as
\begin{equation} \label{QSL-SU2}
    t \geq \frac{1}{\mathcal{V}_\mathds{K}} \arccos\Bigg[\Bigg(1- \frac{3(D-1)}{2(2D-1)} \Bigg)\cos{\sqrt{|\alpha|}t} + \frac{3(D-1)}{2(2D-1)} \Bigg],
\end{equation}
where $\mathcal{V}^2_\mathds{K}=|\alpha|(D+1)/[2(2D-1)]$. As argued above, we expect this bound not to be tight due to the presence of a stationary component in the complexity flow given by the non-vanishing of its trace. We illustrate the deviation from the geodesic trajectory
\begin{equation} \label{geodesic-autocorr}
    \braked{\mathds{K}_0}{\mathds{K}_t}|_{\text{geodesic}}= \|\mathds{K}\|^2\cos{\mathcal{V}_{\mathds{K}}t}
\end{equation}
and the divergence of the two sides of Eq.~\eqref{QSL-SU2} in Fig.~\ref{fig-SU2_withTr}. In what follows, we shall explicitly remove the stationary component of the complexity flow, thus proving the saturation of the refined OQSL \eqref{result3} and showing its equivalence with the dispersion bound \eqref{dispersion-bound}.

\begin{figure}[t]
	\centering
	
	\includegraphics[width=1.01\textwidth]{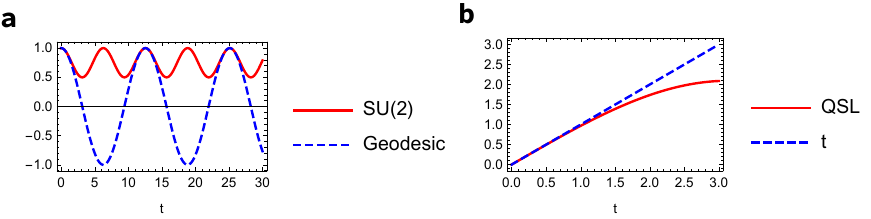}

	\caption{We illustrate the deviation from the saturation of the OQSL \eqref{resultK} for the $SU(2)$ complexity algebra, with Krylov dimension $D=1000$ and $\alpha=-1$. \textbf{a} We compare the evolution \eqref{Kautocorr-SU2} of the complexity autocorrelation function $\braked{\mathds{K}_0}{\mathds{K}_t}$ (red curve) with the geodesic trajectory $ \braked{\mathds{K}_0}{\mathds{K}_t}|_{\text{geodesic}}$ \eqref{geodesic-autocorr} (blue dashed curve). \textbf{b} We show the left and right-hand side of the inequality \eqref{QSL-SU2} (blue dashed and red curves, respectively). The OQSL is tight only near $t=0$, and the deviation increases with time.}
	\label{fig-SU2_withTr}
	
\end{figure}

Finally, although the framework of the geometrical QSL requires a finite dimension, it is instructive to consider the behavior of the quantities involved in Eq.~\eqref{resultK} as $D\to\infty$. In particular, we have already stressed above that the velocity $\mathcal{V}_\mathds{K}$ \eqref{velK} of the complexity flow remains finite, is non-zero for $\alpha>0$ and vanishes for $\alpha=0$. Moreover, from Eq.~\eqref{Kautocorr}, with analogous steps as for the $\textit{SU}(2)$, we obtain that for $\alpha=0$ the autocorrelation function behaves as $\braked{\mathds{K}_0}{\mathds{K}_t} = \|\mathds{K}\|^2 +\frac{\gamma}{2}t^2$. At any finite time, this implies the vanishing of the numerator of the QSL \eqref{resultK}, since the argument of the arccosine approaches $1$ as $D\to\infty$. The evaluation of the limit of the full expression yields as a result that $\tau_{QSL}\to 0$ as $D\to\infty$ in the case of the HW algebra, i.e., for $\alpha=0$. Conversely, the notion of distance employed in our QSL \eqref{result1} is not well defined in the case $\alpha>0$, since the normalized autocorrelation function diverges exponentially with time, $\braked{\mathds{K}_0}{\mathds{K}_t}/\|\mathds{K}\|^2 \sim \exp{\sqrt{\alpha}t}$.

\subsubsection{Saturation of the refined OQSL}

The fact that the Krylov complexity operator $\mathds{K}_t$ cannot saturate the OQSL \eqref{resultK} is already evident from the observation that $\Tr\mathds{K}_t\neq 0$, as this condition results in a non-zero stationary component $\mathds{P}_0$. In particular, the trace accounts for the stationary component along the identity. Removing the trace is sufficient to achieve saturation only if there are no other stationary components of $\mathds{K}_t$, meaning that the zero eigenvalue of the super Liouvillian $\mathds{S}$ has no degeneracy, such that the corresponding eigenspace is spanned by the identity. We explicitly show that this is indeed the case for closed complexity algebras. In this sense, the dispersion bound \eqref{dispersion-bound} and the OQSL \eqref{result3} provide a unique constraint on the operator growth in Krylov space and the saturation of the former automatically implies the saturation of the latter. 

Let us, therefore, consider the operator $\overline{\mathds{K}}_t$ obtained by subtracting from the Krylov complexity its component along the identity:
\begin{equation}
    \overline{\mathds{K}}_t = \mathds{K}_t - \braked{\mathds{K}_t}{\mathds{1}} \frac{\mathds{1}}{\|\mathds{1}\|^2},
\end{equation}
where $\braked{\mathds{K}_t}{\mathds{1}}=\Tr \mathds{K}$ and $\|\mathds{1}\|^2=D$. The OQSL \eqref{result1} for $\overline{\mathds{K}}_t$ reads as
\begin{equation} \label{resultKbar}
    t \geq   \frac{\arccos\Big( \frac{\braked{\overline{\mathds{K}}_0}{\overline{\mathds{K}}_t} }{\|\overline{\mathds{K}}\|^2} \Big)}{\mathcal{V}_{\overline{\mathds{K}}}},
\end{equation}
where $\mathcal{V}^2_{\overline{\mathds{K}}}=\|\mathds{B}\|^2/\|\overline{\mathds{K}}\|^2$. Now, by using Eq.~\eqref{Knorm} and that $\Tr\mathds{K}=D(D-1)/2$, we obtain
\begin{equation}
    \|\overline{\mathds{K}}\|^2=\|\mathds{K}\|^2-\frac{(\Tr\mathds{K})^2}{D}=\frac{D(D^2-1)}{12}.
\end{equation}
Moreover, in the case of $SU(2)$ complexity algebra, from Eq.~\eqref{bn} with $\alpha<0$ and $2\gamma=|\alpha|(D-1)$ we find
\begin{equation}
    \|\mathds{B}\|^2=2\sum_{n=1}^{D-1} b_n^2=|\alpha|\frac{D(D^2-1)}{12}.
\end{equation}
By taking the ratio of the expressions above, we conclude that the velocity of the $\overline{\mathds{K}}_t$ flow for the $SU(2)$ complexity algebra is $\mathcal{V}_{\overline{\mathds{K}}}=\sqrt{|\alpha|}$. On the other hand, the autocorrelation function $\braked{\overline{\mathds{K}}_0}{\overline{\mathds{K}}_t}$ can be computed analogously to the one of $\mathds{K}_t$ in Eqs.~\eqref{Kautocorr}-\eqref{Kautocorr-SU2}, with the only difference that now $\Tr \overline{\mathds{K}}=0$. We thus obtain
\begin{equation} \label{autocorr-bar}
    \braked{\overline{\mathds{K}}_0}{\overline{\mathds{K}}_t}= \|\overline{\mathds{K}}\|^2 \cos\sqrt{|\alpha|}t=\|\overline{\mathds{K}}\|^2 \cos \mathcal{V}_{\overline{\mathds{K}}}t,
\end{equation}
which implies that the inequality in the OQSL \eqref{resultKbar} reduces to an identity at any time. In conclusion, the saturation of the dispersion bound \eqref{dispersion-bound} in finite dimension implies that the Krylov complexity $\mathds{K}_t$ also saturates the refined OQSL \eqref{result3} with $\mathds{P}_0=\mathds{1}\Tr \mathds{K}/D$. We illustrate this saturation in Fig.~\eqref{fig-SU2_withoutTr}. From the comparison between Figs.~\ref{fig-SU2_withTr} and \eqref{fig-SU2_withoutTr} we can assess the efficiency of the refined OQSL \eqref{result3} in yielding a tight evolution by removing the components that do not contribute dynamically to the flow.

\begin{figure}[t]
	\centering
	
	\includegraphics[width=1.01\textwidth]{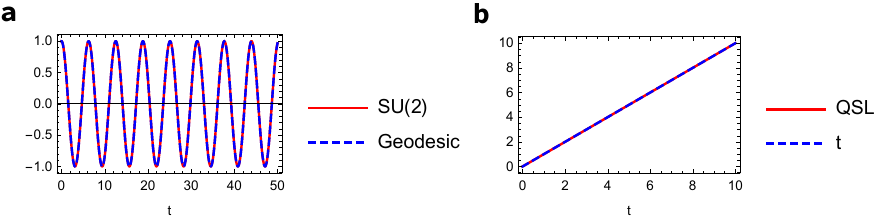}

	\caption{Illustration of the saturation of the OQSL \eqref{resultKbar} for the $SU(2)$ complexity algebra, with Krylov dimension $D=1000$ and $\alpha=-1$. \textbf{a} The complexity autocorrelation function $\braked{\overline{\mathds{K}}_0}{\overline{\mathds{K}}_t}$ \eqref{autocorr-bar} (red curve) is shown to match the geodesic trajectory (blue dashed curve). \textbf{b} The left-hand and right-hand sides of the inequality \eqref{resultKbar} (blue dashed and red curves, respectively). The OQSL reduces to an identity at any time.}
	\label{fig-SU2_withoutTr}
	
\end{figure}

\section{Conclusions, Discussion and Outlook}
Conventional QSLs bound the minimum time for the completion of a process by quantifying the distance traveled by the system along the evolution in state space. OQSLs generalize the scope of conventional QSL to account for processes described in terms of operator flows, i.e., the change of an operator resulting from a conjugation by a one-parameter unitary \cite{Carabba22}. In this work, we have introduced a geometric OQSL that holds for arbitrary unitaries, i.e.,  whether the generator of flow is parameter dependent or not. In addition, we have shown the bound to be tight and identified the required conditions for its saturation. This has led us to introduce a refined OQSL upon identifying the subspace in which the dynamics unfolds.

The usefulness of these OQSLs has been illustrated in the context of a continuous renormalization group, formulated as a Wegner Hamiltonian flow for block diagonalization. In this context, the  flow involves a parameter-dependent generator and its characterization is possible by making use of the geometric OQSL presented. 
We have shown that  Wegner's choice of the flow generator leads to a monotonic decay of the off-diagonal elements of the flowing Hamiltonian towards the target block-diagonal one. However, such a choice does not saturate the OQSL. By contrast, an alternative choice of the generator associated with the Toda flow can lead to the saturation of the OQSL for a specific family of initial Hamiltonians.
Beyond the case of Wegner Hamiltonian flows, we expect our results to apply to other schemes for Hamiltonian diagonalization, such as those relying on the Schrieffer-Wolff transformation \cite{Bravyi11}.

We have further discussed the implication of our results in the context of operator growth in Krylov space. In this representation, the time evolution of an operator is analogous to the spreading of a particle in the Krylov lattice, where the mean position is a proxy for operator complexity. The conditions for maximal operator growth 
are then associated with the saturation of the dispersion bound \cite{Hornedal22}, which occurs when the Lanczos coefficients exhibit a specific dependence on the lattice site index.
Here, we have introduced a ``super-Heisenberg'' representation of the Krylov complexity operator generated by a super Liouvillian. Making use of the OQSL in such representation, we have shown that the saturation of the dispersion bound implies the saturation of the OQSL for the Krylov complexity operator. The application of OQSL to other complexity measures, such as the family of q-complexities including out-of-time-order correlators \cite{Parker2019},  offers an interesting prospect.

Beyond these examples, we expect OQSLs to find manifold applications in 
 the characterization of nonequilibrium phenomena, such as the crossing of a quantum phase transition, the equilibration and thermalization of isolated many-body systems,  quantum thermodynamic processes, quantum control, and quantum annealing. In addition, OQSL may be used in the study of integrable systems, using the zero-curvature representation \cite{Faddeev87},  Lax pairs \cite{Sutherland04}, and Hamiltonian deformations \cite{Gross20,Gross20b,Matsoukas-Roubeas22}. In addition,  our results apply directly to  isospectral operator flows with a double bracket structure appearing in projected gradient methods for least-square matrix approximations \cite{ChuDriessel90}, and dynamical systems for list sorting and linear programming problems \cite{Brockett91}. 
The generalization of our results to dissipative quantum systems would be highly desirable, given its prospective applications, e.g.,  to the description of open quantum dynamics in Heisenberg's representation and the quest for fundamental limits to nonunitary operator growth \cite{Bhattacharya22,LiuTangZhai22,Bhattacharjee22b}. 

\section{Acknowledgements}
It is a pleasure to acknowledge discussions with  L\'eonce Dupays,  Anatoly Dymarsky,  
\'Iñigo L. Egusquiza, and Federico Roccati. KT acknowledges support by 
JSPS KAKENHI grant No. JP20K03781 and No. JP20H01827.


\appendix
\section*{Appendices}
\addcontentsline{toc}{section}{Appendices}
\renewcommand{\thesubsection}{\Alph{subsection}}

\subsection{Proving bijection between positive semi-definite inner-products and positive semi-definite operators}
\label{app-A}

\begin{lem}
    Let $\P\in\textrm{End}(\B)$ be a positive semi-definite superoperator. The binary operation $\inner{\cdot}{\P\cdot}_\textsc{h} : \B\times\B \rightarrow \mathds{R}$ is a positive semi-definite inner product on $\B$.
\end{lem}
\begin{proof}
    We need to show that the map $\inner{\cdot}{\P\cdot}_\textsc{h}$ satisfies linearity, Hermitian symmetry, and positive semi-definiteness.
    
    \underline{\textbf{Linearity}}
    
    Consider any triplet of operators $A$, $B$ and $C$ and a complex number $\lambda$. We then have
    \begin{align}
        &\inner{A}{\P\lambda B}_\textsc{h} = \inner{A}{\lambda\P B}_\textsc{h} = \lambda\inner{A}{\P B}_\textsc{h}\\
        &\inner{A}{\P(B+C)}_\textsc{h} = \inner{A}{\P B+\P C }_\textsc{h} = \inner{A}{\P B}_\textsc{h} + \inner{A}{\P C }_\textsc{h},
    \end{align}
    where we have used the linearity of the Hilbert-Schmidt inner product and the superoperator.
    
    \underline{\textbf{Hermitian symmetry}}
    
    For any pair of operators $A$ and $B$ we have
    \begin{equation}
        \inner{A}{\P B}_\textsc{h} = \inner{\P B}{A}_\textsc{h}^* = \inner{B}{\P^\dagger(A)}_\textsc{h}^* = \inner{B}{\P A}_\textsc{h}^*,
    \end{equation}
    where we have used the Hermitian symmetry property of the Hilbert-Schmidt inner product and the fact that a positive semi-definite superoperator is self-adjoint.
    
    \underline{\textbf{Positive semi-definiteness}}
    
    From the definition of positive semi-definiteness of a superoperator 
    it follows directly that
    \begin{equation}
        \inner{A}{\P A}_\textsc{h} \geq 0,
    \end{equation}
    for any operator $A$.
\end{proof}
\begin{prop}
    Given the Hilbert space $(\B, \inner{\cdot}{\cdot}_\textsc{h})$, the map $\P \mapsto \inner{\cdot}{\P \cdot }_\textsc{h}$ is a bijection between the set of positive semi-definite operators on $(\B, \inner{\cdot}{\cdot}_\textsc{h})$ and the set of positive semi-definite inner products on $\B$.
\end{prop}

\begin{proof}
    We can prove that the map is a bijection if we can prove that it is injective and surjective.
    
    \underline{\textbf{Injectivity}}
    
    Suppose that $\inner{\cdot}{\P \cdot }_\textsc{h} = \inner{\cdot}{\eta(\cdot)}_\textsc{h}$ for some pair of superoperators $\P$ and $\eta$. It follows that
    \begin{align}
    \begin{split}
        \inner{\cdot}{\P \cdot }_\textsc{h} = \inner{\cdot}{\eta(\cdot)}_\textsc{h} &\iff \inner{A}{\P B}_\textsc{h} = \inner{A}{\eta(B)}_\textsc{h}\quad\forall A,B\in\B\\
        &\iff \P = \eta.
        \end{split}
    \end{align}
    \underline{\textbf{Surjectivity}}
    
    Given any positive semi-definite inner product $\braked{\cdot}{\cdot}$, we want to construct a positive semi-definite superoperator $\P$ such that $\braked{\cdot}{\cdot} = \inner{\cdot}{\P \cdot }_\textsc{h}$. Let $M_1, M_2 \dots M_{n^2}$ be an operator basis in $\B$ and $a_k$ and $b_k$ be the corresponding components of $A$ and $B$. Define $\inner{M_i}{\P M_j }_\textsc{h} = \braked{M_i}{M_j}$. This superoperator is positive semi-definite since
    \begin{align}
    \label{pos-def}
    \begin{split}
        \inner{A}{\P B}_\textrm{H} &= \sum_{i=1}^{n^2}\sum_{j=1}^{n^2}a^*_ib_j\inner{M_i}{\P M_j }_\textrm{H}
        = \sum_{i=1}^{n^2}\sum_{j=1}^{n^2}a^*_ib_j\braked{M_i}{M_j}\\
        &= \braked{A}{B} \quad \implies \quad  \inner{A}{\P A}_\textrm{H} = \braked{A}{A} \geq 0,
        \end{split}
    \end{align}
    from which it is also clear that $\P$ maps to $\braked{\cdot}{\cdot}$.
\end{proof}

\subsection{The kernel of seminorms}
\label{app-B}

\begin{prop}
\label{kern}
    Let $\norm{\cdot}$ be the seminorm induced by the positive semi-definite inner product $\inner{\cdot}{\P \cdot }_\textsc{h}$. It is then the case that $\norm{A} = 0 \iff\P A = 0$.
\end{prop}

\begin{proof}
    Let $\{M_k\}_{k=0}^{n^2}$ be an orthonormal eigenbasis of $\P$ such that $\lambda_k$ is the eigenvalue corresponding to the eigenvector $M_k$. Let $a_k$ be the components of an operator $A$. We then have
    \begin{align}
    \begin{split}
        \inner{A}{\P A}_\textrm{H} &= \sum_{i=1}^{n^2}\sum_{j=1}^{n^2}a^*_ia_j\inner{M_i}{\P M_j }_\textrm{H} = \sum_{i=1}^{n^2}\sum_{j=1}^{n^2}a^*_ia_j\lambda_j\delta_{ij}\\
        &= \sum_{k=1}^{n^2}\abs{a_k}^2\lambda_k.
    \end{split}
    \end{align}
    Since $\lambda_k \geq 0$, this sum can only be zero if $a_k = 0$ for $\lambda_k > 0$, in other words, $A$ lies in the kernel of $\P$.
\end{proof}

\subsection{Proof of equation \ref{inner-equality}}
\label{app-C}

\begin{proof}
    Define $\Tilde{A} = A-\hat{A}$. We can choose an orthogonal basis $M_1, M_2,\dots M_d$,  $N_1, N_2,\dots N_{n-d}$, such that the operators $M_k$ and $N_k$ span $\textrm{im}(\P)$ and $\textrm{ker}(\P)$ respectively and $d$ is the dimension of $\textrm{im}(\P)$. Let $a_k$ be the components of $A$ with respect to this basis. We have that
    \begin{align}
    \begin{split}
    \P \Tilde{A}  =
        \P(A-\hat{A}) &= \P A - \P \hat{A}\\
        &= \sum_{k=1}^d a_k\P M_k  + \sum_{k=1}^{n-d} a_k\P N_k - \sum_{k=1}^d a_k\P M_k  = 0.
    \end{split}
    \end{align}
    Using this together with linearity of $\braked{\cdot}{\cdot}$, we get
    \begin{equation}
        \braked{A}{B} = \braked{\hat{A}}{\hat{B}} + \braked{\hat{A}}{\Tilde{B}} + \braked{\Tilde{A}}{\hat{B}} + \braked{\Tilde{A}}{\Tilde{B}} = \braked{\hat{A}}{\hat{B}} = \braket*{\hat{A}}{\hat{B}},
    \end{equation}
    where the step to the second equality follows from proposition \ref{kern} and the last step follows from the definition of $\inner{\cdot}{\cdot}$.
\end{proof}

\subsection{Smallest subspace containing the dynamics}
\label{app-D}

Let $V$ be a real or complex finite dimensional vector space and let $L$ be a linear endomorphism on $V$. Assuming that $L$ is diagonalizable, we can write $L = \sum_{i=1}^d l_iP_i$, where $l_i$ are the $d$ distinct eigenvalues of $L$ and $P_i$ are the projections onto the corresponding eigenspaces satisfying $\sum_{i=1}^d P_i = I$ and $P_iP_j = \delta_{ij}P_i$. Here $I$ is the identity map on $V$. It follows from the definition of the exponential function that $e^{Lt} = \sum_{i=1}^d e^{l_it}P_i$. Consider now any initial vector $v$ in $V$ evolving according to $v(t) = e^{Lt}v$. Let us define the subspace $W = \textrm{span}\{v_i\}_{i\in I}$ where $i\in I \iff v_i := P_iv \neq 0$. We then have that $v(t) = \sum_{i\in I} e^{l_i t}v_i$ and we see that the evolution is entirely contained in $W$. Given a proper time interval $T\subset \mathbb{R}$, we now ask whether $W$ is the smallest subspace for which \{$v(t)$ : $t\in T$\} is contained in.\footnote{A proper interval is an interval in $\mathbb{R}$ excluding the empty set and singletons.} The answer is affirmative. To show this, first, note that the functions $e^{l_it}$ with domain $T$ are linearly independent given that all eigenvalues $l_i$ are distinct. This implies that 
\begin{equation}
\label{vector2}
    \sum_{i\in I}c_ie^{l_i t} = 0\textrm{  }\forall t\in T\iff c_i = 0\textrm{  }\forall i\in I,
\end{equation}
where $c_i\in\mathbb{C}$. We will use proof by contradiction to show that $W$ must be the smallest subspace. Assume that there exists a subspace $\Tilde{W}$ containing the evolution with a dimension strictly smaller than $W$. We must then have the evolution contained in the subspace given by the intersection $F = W \cap \Tilde{W}$. By our assumption, $F$ must then have a dimension strictly smaller than $W$. This implies that there exists a non-zero linear functional $w$ with domain $W$ for which $F\subset\textrm{ker}(w)$. We can expand this functional in the basis $\{f_i\}_{i\in I}$ defined by $f_i(v_j) = \delta_{ij}$ so that $w = \sum_{i\in I}w_i f_i$, where $w_i = w(v_i)$. We get that $w\big(v(t)\big) = 0\textrm{  }\forall t\in T \iff \sum_{i\in I}w_ie^{l_i t} = 0\textrm{  }\forall t\in T$. This last expression together with \eqref{vector2} implies that $w_i = 0$ $\forall i\in I\iff w = 0$. Thus, we have reached a contradiction; hence, $W$ is the smallest subspace containing \{$v(t)$ : $t\in T$\}.

The proof can be carried out analogously for the case when $L$ is time dependent but commutes, i.e., $[L(t_1),L(t_2)] = 0$ $\forall t_1,t_2\in T$.

\subsection{Optimal refinement}
\label{app-E}

Consider the decomposition $\hat{A} = S + V_t$ discussed in section \ref{refined} and assume that the subspace $\ker(\mathds{L})\cap\H_\P$ does not change over time. Since $S\in \H_\P$ does not evolve per assumption under the influence of $\mathds{L}$, it must also belong to $\ker(\mathds{L})\cap\H_\P$. Suppose $V_t$ has a non-zero projection $S'$ onto $\ker(\mathds{L})\cap\H_\P$ then $S'$ is guaranteed to remain unchanged since we have by assumption that $\ker(\mathds{L})\cap\H_\P$ is time-independent. Let $V_t'$ be the orthogonal complement so that $V_t = S' + V_t'$. It follows that
\begin{equation}
    \inner{V}{V_t} = \braked{V}{V_t} = \braked{S'+V'}{S'+V'_t} = \braked{V'}{V'_t} + \norm{S'}^2,
\end{equation}
where we have used the fact that $S'$ is orthogonal to $V'_t$, $V'\equiv V'_0$. This implies that
\begin{equation}
    \Re \braked{V'}{V'_t} = \Re \braked{V}{V_t} - \norm{S'}^2 = \Re C(t) - \norm{S}^2 - \norm{S'}^2 = \Re C(t) - \norm{S+S'}^2,
\end{equation}
where the last equality follows from the assumption $\Re \inner{S}{V_t} = 0$. We can thus improve the speed limit further whenever $S'\neq 0$ since we would then have that $\norm{S+S'}=\norm{S} + \norm{S'} > \norm{S}$. The operator $P_0 = S+S'$ is precisely the orthogonal projection of $\hat{A}_t$ onto $\ker(\mathds{L})\cap\H_\P$.

\subsection{Proving orthogonality from the preservation of norm}
\label{app-F}

We here want to show that the operators $P_0$, $X_\omega$ and $Y_\omega$ in \eqref{optimal dynamics} are orthogonal and that $X_\omega$ and $Y_\omega$ have the same norm. Using \eqref{optimal dynamics}, we get that
\begin{align}
\label{X}
    \norm{A_t}^2 &= \norm{P_0}^2 + \norm{P_\omega}^2 + \norm{P_{-\omega}}^2 +2\cos{\theta(t)}\Re\inner{P_0}{X_\omega} + 2\cos(2\theta(t))\Re\inner{P_\omega}{P_{-\omega}}\\
    \label{Y}
    &= \norm{P_0}^2 + \norm{P_\omega}^2 + \norm{P_{-\omega}}^2 +2\sin{\theta(t)}\Re\inner{P_0}{Y_\omega} + 2\cos(2\theta(t))\Re\inner{P_\omega}{P_{-\omega}}.
\end{align}
At time $t=0$, we have that $\theta(0) = 0$ and we get from expression \eqref{X} and \eqref{Y} that $\Re\inner{P_0}{X_\omega} = 0$. Assuming that $\theta(t)$ is not zero over the whole interval $[0,\tau]$, we can conclude from $\Re\inner{P_0}{X_\omega} = 0$, \eqref{X} and \eqref{Y} that $\Re\inner{P_0}{Y_\omega} = 0$ must also hold. This in turn implies that $\Re\inner{P_\omega}{P_{-\omega}} = 0$ since $\norm{A_t}$ must be constant. This last equality guarantees that the norm of $X_\omega$ and $Y_\omega$ are equal. Note that in the case when $\theta(t) = 0$ over the whole interval, we have that the dynamics is stationary.

\nocite{apsrev41Control}
\bibliography{Quantum.bib} 
\bibliographystyle{apsrev4-1}
\end{document}